\documentclass[11pt]{article}
\usepackage[letterpaper,margin=1in]{geometry}
\newif\ifdraft
\draftfalse
\usepackage{amsmath, amssymb,amsthm}
\usepackage{mathtools} 
\usepackage[utf8]{inputenc}
\usepackage[T1]{fontenc}
\usepackage{float}
\usepackage{dsfont}

\usepackage{color}
\usepackage{cuted}
\usepackage{enumerate}

\usepackage{placeins}

\makeatletter
\newtheorem*{rep@theorem}{\rep@title}
\newcommand{\newreptheorem}[2]{%
  \newenvironment{rep#1}[1]{%
    \def\rep@title{\textbf{#2 \ref{##1}}}%
    \begin{rep@theorem}}%
    {\end{rep@theorem}}}
\makeatother

\newreptheorem{theorem}{Theorem}
\newreptheorem{thm}{Theorem}
\newreptheorem{lem}{Lemma}
\newreptheorem{defn}{Definition}
\newreptheorem{lemma}{Lemma}
\newreptheorem{proposition}{Proposition}
\newreptheorem{corollary}{Corollary}

\usepackage{xspace}
\usepackage{xcolor}
\definecolor{darkred}{rgb}{0.5,0,0}
\definecolor{lightblue}{rgb}{0,0.4,0.8}
\definecolor{darkgreen}{rgb}{0,0.5,0}
\definecolor{grey}{rgb}{0.5, 0.5, 0.5}
\usepackage{hyperref}
 \hypersetup{%
     colorlinks=true, linktocpage=true, pdfstartpage=1, pdfstartview=FitV,%
     breaklinks=true, pdfpagemode=UseNone, pageanchor=true, pdfpagemode=UseOutlines,%
     plainpages=false, bookmarksnumbered, bookmarksopen=true, bookmarksopenlevel=1,%
     hypertexnames=true, pdfhighlight=/O,
     urlcolor=darkred, linkcolor=lightblue, citecolor=darkgreen, 
     pdftitle={},%
     pdfauthor={},%
     pdfsubject={},%
     pdfkeywords={},%
     pdfcreator={pdfLaTeX},%
     pdfproducer={LaTeX}%
 }

\usepackage[style=alphabetic,backend=biber,natbib=true,hyperref=true,
maxbibnames=99,
maxcitenames=4,
mincitenames=1,
maxalphanames=4,
minalphanames=3,
giveninits=true,
date=year,
sorting=nyt,
bibencoding=auto,
doi=false,isbn=false,
]{biblatex}

\bibliography{biblio}

\usepackage{xhfill}
\usepackage{relsize}
\usepackage[section]{algorithm}
\usepackage[noend]{algorithmic}

\usepackage{graphicx}
\usepackage{subcaption}

\usepackage{hyphenat}
 
\newcommand{\poly}{\text{poly}}

\newcommand{\eps}{\varepsilon}

\newcommand{\calE}{\mathcal{E}}

\newcommand{\keq}{\kappa\xspace} 
\newcommand{\ksg}{\lambda}
     
\newcommand{\Topk}{\operatorname{\bf Top}}

\newcommand{\Min}{\operatorname{\bf ParallelMin}}

\newcommand{\Max}{\operatorname{\bf ParallelMax}}
\newcommand{\approxtopk}{\operatorname{\bf Approx--}\Topk}
\newcommand{\Threshold}{\operatorname{\bf Threshold-v}}

\newcommand{\OneRoundMax}{\operatorname{\bf OneRoundMax}}

 \usepackage{aliascnt}
\def\NewTheorem#1#2{%
  \newaliascnt{#1}{theorem}
  \newtheorem{#1}[#1]{#2}
  \aliascntresetthe{#1}
  \expandafter\def\csname #1autorefname\endcsname{#2}
}

 \newtheorem{theorem}{Theorem}[section]
\NewTheorem{lemma}{Lemma}
\NewTheorem{corollary}{Corollary}
\NewTheorem{conjecture}{Conjecture}
\NewTheorem{observation}{Observation}
\NewTheorem{assumption}{Assumption} 

\NewTheorem{definition}{Definition}
\NewTheorem{proposition}{Proposition}
\NewTheorem{remark}{Remark}

\NewTheorem{claim}{Claim}
\newtheorem{notation}{Notation}

\theoremstyle{remark}



\renewcommand{\Pr}[1]{\mathbb{P}\left(\,#1\,\right)}
\newcommand\E[1]{\mathbb{E}\left[\,#1\,\right]}

\newcommand{\makenote}[2]{{{\color{#1} #2}}}
\newcommand{\fnote}[1]{
	\ifdraft 
\makenote{blue}{FM: #1}\fi}
\newcommand{\silent}[1]{}

\newcommand{\claire}[1]{	\ifdraft \makenote{darkgreen}{Claire: #1} \fi}
\newcommand{\vincent}[1]{	\ifdraft \makenote{magenta}{Vincent: #1} \fi}

\newcommand{\inparallel}{\emph{in parallel}\xspace}
\newcommand{\Output}{\textbf{output}\xspace}

\newcommand{\ie}{{\it i.e.,}\xspace}
\newcommand{\eg}{{\it e.g.}\xspace}

\newcommand{\etc}{{\it etc.}\xspace}

\title{Instance-Optimality in the Noisy Value-and Comparison-Model}
\begin{document}

\author{Vincent Cohen-Addad \\ CNRS \& Sorbonne Universit\'e\\
\texttt{vcohenad@gmail.com}
\and Frederik Mallmann-Trenn\thanks{
This work was supported in part by NSF Award Numbers  CF-1461559, CCF-0939370, and CCF-1810758.}\\
MIT\\
\texttt{mallmann@mit.edu}
\and Claire Mathieu \\ CNRS\\
\texttt{claire.m.mathieu@gmail.com}}

\date{}

\maketitle
\thispagestyle{empty}

\begin{abstract}
Motivated by crowdsourced computation, peer-grading, and recommendation systems, Braverman, Mao and Weinberg [STOC'16] 
studied  
the \emph{query} and \emph{round} complexity of fundamental problems such as 
finding the maximum (\textsc{max}), finding all elements above a certain value (\textsc{threshold-$v$})
or computing the top$-k$ elements  (\textsc{Top}-$k$) in a noisy environment.

For example, consider the task of selecting papers for a conference.
This task is challenging 
due the crowdsourcing nature of peer reviews:
 the results of reviews are noisy and it is necessary
to parallelize the review process as much as possible.
We study the noisy value model and the noisy comparison model:
In the \emph{noisy value model}, a reviewer is asked to evaluate a single element:  ``What is the value of paper $i$?''  (\eg accept).
In the \emph{noisy comparison model} (introduced in the seminal work of Feige, Peleg, Raghavan and Upfal [SICOMP'94]) 
a reviewer is asked to do a pairwise comparison: ``Is paper $i$ better than paper $j$?''

In this paper, we show optimal worst-case query complexity  for the \textsc{max},\textsc{threshold-$v$} and \textsc{Top}-$k$ problems.
 For \textsc{max} and \textsc{Top}-$k$, we obtain optimal worst-case upper and lower bounds on the round vs query complexity
 in both models. For \textsc{threshold}-$v$, we obtain
optimal query complexity and nearly-optimal round complexity (\ie optimal up to a factor $O(\log \log k)$, 
where $k$ is the size of the output)  for both models.

We then go beyond the worst-case and address the question of the importance of knowledge of the instance by providing,
for a large range of parameters, instance-optimal algorithms with respect to the query complexity.
 We complement these results by showing that for some family of instances, no instance-optimal algorithm can exist. 
Furthermore, we show that the value-and comparison-model 
are for most practical settings asymptotically equivalent (for all the above mentioned problems); 
on the other hand, in the special case where the papers are totally ordered, we show that the value model 
 is strictly easier than the comparison model.

\end{abstract}

\newpage
\thispagestyle{empty}

 \tableofcontents 

%
%
%

\newpage
\section{Introduction}
%
%
\setcounter{page}{1}

Computing with noisy information is a fundamental problem in
computer science. Since the seminal work of
Feige, Peleg, Raghavan and Upfal~\cite{FRPU94},
there have been a variety of algorithmic results  in the context of noisy operations, arising from both practitioners and theoreticians, on
classic problems such as finding the maximum (\textsc{max}) or
the top-$k$ elements (\textsc{top-$k$})~\cite{FRPU94,BMW16,chen2015spectral,busa2013top,eriksson2013learning}.

There has recently been renewed interest for those
problems, motivated by problems arising in rank aggregation,
crowdsourcing,
peer-grading and recommendation
systems. In such systems, a user may be asked to give
a grade (for example using the popular ``five-star'' rating), a \emph{value query};
or they may be asked to perform a
comparison between two data elements, a \emph{comparison query} .
In addition to  finding the maximum or the top-$k$ elements, 
these systems may also aim to find all elements with values greater than some threshold value $v$
(\textsc{threshold}-$v$). The problem of finding the top-$k$ elements
of a list of distinct elements, \textsc{rank-$k$} (a special case of \textsc{top-$k$}),
has also been widely studied (\eg~\cite{GS10,newman2004computing,fagin}). 
The outcome of any such query is very noisy. The most basic noise model 
\citet{FRPU94} assumes that queries fail 
independently with some constant probability.

In addition, rounds of interactions between
users making queries and the system that selects the queries---that should be
performed next---are costly in practice. For instance, Braverman
et al.~\cite{BMW16}
pointed out that communicating to the users the new tasks they must
perform creates a computational bottleneck.
It is thus crucial for the
application to deal with noisy queries and to
minimize
both the total number of queries and the total number
of rounds of interactions.

Hence, Braverman et al.~\cite{BMW16} considered
the \emph{round complexity}
(also, for example~\cite{GS10} and~\cite{newman2004computing} 
before them)
of the maximum and the \textsc{top-$k$} and \textsc{rank-$k$} problems
in different models of comparisons 
and in particular in the noisy comparison model. An algorithm has round complexity $r$ if its queries can be partitioned into a sequence of $r$ batches, where the queries of a batch only depend on the answers received for the previous batches.
Braverman et al.~\cite{BMW16} provided lower bounds on the number of rounds
of interactions that are needed in order to find
the largest element of a set when comparisons 
can fail with constant probability. See \autoref{S:related} for 
more details.
\smallskip

\paragraph{Our results.} In this paper, we provide a careful analysis of the round and query
complexity of \textsc{max}, \textsc{top-$k$}, \textsc{threshold}-$v$, and \textsc{rank-$k$} 
in both the worst-case and instance-optimal scenarios for both the noisy comparison and noisy value models.
We provide lower and upper bounds for the trade-off between round and query complexity. Our algorithms 
for \textsc{max}, \textsc{top-$k$}, and \textsc{bounded-threshold-$v$}\footnote{Here the number of distinct values is bounded by $n^{1-\varepsilon}, \varepsilon>0$ which is the case in many practical settings such as rating movies, papers, \etc.} are optimal w.r.t. the query complexity 
and optimal up to a factor $\log \log k$ w.r.t. the round complexity, where $k$ is the size of the output. 
This is a significant improvement
over previous work.
As a byproduct, we show that  the noisy value model and the noisy comparison model are in 
many settings of interest  essentially equivalent w.r.t. the query complexity. 
On the other side, we show a separation between the two models 
for the \textsc{rank-$k$} problem.


We go one step further and give 
fine-grained upper and lower bounds on the query complexity
through the classic notion
of \emph{instance-optimality} (see \autoref{sec:prelims}
for a formal definition). The algorithm frequently has additional information about some features of the input; for example, it might have, from prior experience, an estimate of the average value of the items, or some information about the distribution of those values. However, the address at which the elements are stored in memory bears no connection to their values, so, even if the algorithm has full knowledge of the set (or multiset) of values given as input, it still has the task of finding where the values of interest are stored. An algorithm is instance-optimal if it has no prior knowledge of the distribution, yet has complexity equal, up to a constant factor, to that of the best algorithm with full prior knowledge of the set (or multiset) of values.  For several problems, we provide algorithms that are instance-optimal with respect to query complexity---\ie showing that knowing the instance in advance, up to a permutation, yields no benefit.
We complement these results by showing that for some family of instances, no instance-optimal algorithm exists. 
%
In more detail, we show the following:

\begin{itemize}
\item \textbf{Worst-Case Bounds for \textsc{max},  \textsc{threshold}-$v$, \textsc{top-$k$}, and  \textsc{rank-$k$}:} 
As a first step, we give optimal
  bounds for the round vs query complexity on
  \textsc{max} (\autoref{lem:topksidekicks}
  and~\autoref{theorem:maxLB}),  
  in both noisy-value and noisy-comparison models.
  
  For the \textsc{threshold}-$v$ problem, we provide optimal
  bounds on the query complexity and nearly-optimal 
  bounds on the round complexity
  (optimal up to a factor $O(\log \log k)$ where
  $k$ is the size of the output) in the value-queries
  model\footnote{Note that the \textsc{threshold}-$v$
    problem is not well-defined
    in the comparison model} (\autoref{pro:mightyUB}
  and~\autoref{pro:mightyLB}). 

  For \textsc{top-$k$},
  we show that any algorithm
  with success probability $2/3$ that finds the \textsc{top-$k$}
  elements of an $n$ elements set in $r$ rounds
  requires $\Omega(n \log (kb))$ queries, where
  $b$ satisfies $r > \log^*_{bk}(n) - 4$
  (\autoref{thm:kmaxlower}).
  We provide an optimal algorithm with query 
  complexity $O(n \log(kb))$ 
  and round complexity $r+O(1)$
  (\autoref{thm:topksidekicks}).
  
  Similar bounds hold for 
  \textsc{rank-$k$} in the comparison model (\autoref{thm:topksidekicks}).
  This generalizes the 4-round $O(n \log n)$ algorithm
  of \citet{BMW16}.
  On the other hand, the query complexity is $O(n+k \log n)$ in the value model   (\autoref{thm:rankkupper}).
  Our  lower bounds give the first 
  trade-off between round and query complexity for the problems. 
  See \autoref{tbl:querycomplexity} for a summary.

    \FloatBarrier

    \hspace{2cm}
  \begin{table}[!htbp]
   \centering 
\begin{tabular}{|l|l|l|}
\hline
                     & \textbf{Query Complexity} & \textbf{References} \\ \hline
\textsc{Max}         & $\Theta(n)$   &       folklore,     \autoref{lem:topksidekicks}, \autoref {theorem:maxLB}       \\ \hline
\textsc{threshold}-$v$  &    $\Theta(n \log k_v)$                 &        \autoref{pro:mightyUB}, \autoref{pro:mightyLB}                 \\ \hline
\textsc{top-$k$}     &        $\Theta(n\log k)$                   &       \autoref{thm:topksidekicks}, \autoref {thm:kmaxlower}              \\ \hline
\textsc{rank-$k$}  (comparison model)  &           $\Theta(n\log k)$                       &       \autoref{thm:topksidekicks},        \autoref{thm:uniquelower}      \\ \hline
\textsc{rank-$k$}  (value model)  &           $O(n+k\log n)$                       &       \autoref{thm:rankkupper}        \\ \hline 
\end{tabular}
\caption{Query complexity (upper and lower bounds on maximum number of queries) to ensure correctness probability at least $2/3$.  For \textsc{Max} and \textsc{top-$k$}  the bounds hold in both the value model and the comparison model. For \textsc{threshold}-$v$,  $k_v$ denotes the number of elements with value at least $v$, and the bounds are for the value model (the problem is not well-defined in the comparison model). We assume $1\leq k,k_v\leq n/2$ for better readability. }
\label{tbl:querycomplexity}
\end{table}
  \FloatBarrier

\item \textbf{Approximation Algorithms:} 
  To bypass our lower bounds for the  worst-case scenario and remove the dependency on $k$, we initiate the study of approximation algorithms
  for these problems.
  For the \textsc{top-$k$} problem, we provide an algorithm
  that returns $k$ elements
  among the top-$(1+\epsilon)k$ elements using
  $O(\log_{{1}/{\varepsilon}}^* n)$ 
  rounds and $O(n \log(\frac{1}{\eps}))$ queries
  with probability $2/3$.
  Our lower bounds extend to this setting and show that
  this bound is tight, for a large range of parameters.

  
  
\item \textbf{Instance-Optimal Bounds:}
    For the \textsc{Max} problem,
  we show that any algorithm that has prior knowledge of the
  instance except for the actual permutation of the input
  elements (we call such an algorithm a non-oblivious algorithm)
  and that finds the unique maximum with success probability
  at least $2/3$ requires $\Omega(n)$
  queries. We provide an oblivious algorithm (\ie with no
  prior knowledge of the instance) whose query complexity, $O(n)$, 
  matches the same bound (\autoref{lem:topksidekicks}
  and~\autoref{theorem:maxLB}).

   For the \textsc{threshold}-$v$ problem we provide 
  an oblivious algorithm,
  that makes $O(n \log k)$ queries (where here again $k$ is the
  size of the output).
  The round complexity is $O(\log \log k \cdot \log^*n)$
  (\autoref{cor:threshold-v:instanceopt}).
  We show that in this general setting, no oblivious instance-optimal
  algorithm exists (\autoref{thm:nofreelunch}).
  %
  We observe that in several applications---such as peer-review processes,
  grading students or 
  evaluating the  quality of a service---the grades can take a constant 
  number of values. Thus, we consider the problem of
  identifying all the elements with value greater than $v$ in an
  $n$ elements set whose values are taken in $[\ell]$,
  $\ell = O({n}^{1-\varepsilon})$, which we call
  the \textsc{bounded-value threshold-$v$} problem.
  We  show that in this setting
  any non-oblivious algorithm with success probability at least $2/3$
  requires                   
  $\Omega(n \log k)$
  queries, where $k$ is the output size (\autoref{pro:mightyLB}).
  This shows that our oblivious algorithm is instance-optimal
  for the \textsc{bounded-value threshold-$v$} problem.

  We then provide a careful analysis of the instance-optimal complexity
  of the classic \textsc{top-$k$} problem.
  We show for a large range of parameters for $k,\kappa$ and $\ell$ the following results:
  \begin{itemize}
  \item Any non-oblivious algorithm that solves the \textsc{top-$k$} problem with probability at least
  2/3, and with an unbounded number of rounds, requires at least $\Omega(n \log (\lambda + \kappa/(s+1)))$
  queries, where $\lambda$ is the number of elements whose values are greater than the $k$th value, $\kappa$
  is the number of elements whose values are equal to the $k$th value and $s = \lambda + \kappa - k$ (\autoref{theorem:LB}).
  \item We provide a nearly-instance-optimal oblivious algorithm; We give an algorithm with query
  complexity $O(n \log (\lambda + \kappa/(s+1)) + k \log k)$ and round complexity
             $O(\log^*n \cdot \log \log (\lambda + \kappa/(s+1)))$ (\autoref{theorem:obliviousalgorithm}).
  \end{itemize}


  


\item \textbf{Separation Between the Value- and the Comparison Model:}
  We consider the\textsc{rank}-$k$ problem in the value and comparison
  models. Interestingly, we show a separation between the two 
  models; the value model is strictly easier than the comparison
  model.
  For the value model, we give a  algorithm making
  $O(n+k\log n)$ queries (\autoref{thm:rankkupper}).
  For the comparison model, we give an instance-optimal
  algorithm making $\Theta(n \log k)$ queries
  and show that this is tight
  (\autoref{thm:uniquelower}), no matter the number of rounds. 

\end{itemize}


For example, the practical implications of our work for the review process, 
in our model, are as follows:
Finding the \textsc{top-$20$} papers can be done in
$r$ rounds, and the number $q$ of reviews is at most 
$O( n\underbrace{\log\log\cdots \log}_{r} n)$
and this is optimal.
On the other hand, if one seeks to find the top-$k$ papers among the 
top-$k(1+\varepsilon)$ papers, for constant $\varepsilon>0$, then we show how this can be done in
average constant number of queries (reviews) per paper using 
$\log^* n+O(1)$ review rounds.

\paragraph{Technical Contributions.}
Our algorithms are  simple  and easy to implement, even in a distributed
environment. They rely on carefully designed divide-and-conquer procedures and build upon each other. One particularly
interesting algorithm is $\approxtopk$. It has an optimal query complexity despite having to estimate
two crucial parameters simultaneously. 

Nonetheless, our main contributions are our lower bounds. To understand
where
our contributions lie, consider the following classical approaches
for the design of lower bounds in the comparison model (e.g., \cite{FRPU94}).
The authors design two worst-case instances such that
any algorithm has to distinguish between them in order to find a correct output.
By design of the instances, this requires many queries.

This approach does not extend to  instance-optimal lower bounds since (1) any algorithm that has prior knowledge of the
instance does not have to pay the price to distinguish the instance at
hand from any other instance\footnote{Compare to our definition of instance optimality \autoref{def:ionew}.}\fnote{can someone add a meaningful sentence?},
and (2) the lower bound should hold for any instance, not only
a particularly hard instance. Designing
instance-optimal lower bound thus becomes a harder challenge: the lower bound on the
number of queries should now come from the problem of identifying the correct output among the
input set of elements $X$. Moreover, for this lower bound to apply
to \emph{any} instance, the hardness
has to hold no matter what the structure of
X is (there could for example be multiple elements sharing the same
value, etc.).

To bypass this barrier, we make use of several ingredients.
The first step,  is is to move to a
more general setting which forces the algorithm to work in
many \emph{phases} (generalizing the two-phase approach of inspired by~\cite{FRPU94}).
After each phase we characterize (probabilistically) the knowledge of the algorithm;
to characterize the knowledge in a compact way, to do so, we make use of
various tools such as for example the 
``little-birdie'' principle. 
\fnote{remove?} \vincent{changed, what do you think?}

Our goal is to show that if the algorithm does not
make enough queries, it does not have enough information to identify
the elements of the output among the entire input set: with decent probability there are multiple elements for which the algorithm has received exactly the same information and is hence forced to guess the correct output.

 

Another major challenge is that the lower bound has to hold for
any instance. For non-instance-optimal lower bounds, it is enough
to define a family of instances and give the adversary's strategy (``lying scheme'')
on these instances. For instance-optimal lower bounds one has to
design an efficient strategy for the adversary for each instance.
We achieve this thanks to the precise analysis of the information the
algorithm has during each phase. 
\medskip
In addition to developing new lower bounds techniques in the value model,
we also generalize
lower bounds in the comparison models.
Our approach is to generalize the technique of \cite{BMW16}, which was developed to show
a lower bound on the query complexity of finding the maximum element, to hiding $k$ elements. 
The core-idea of the approach in \cite{BMW16} is to consider the comparison tree $L$ and to consider the joint distribution
of given input permutation and $L$: For every correct output, there are many likely input permutations that would have yielded an identical output,
which would be incorrect. We show how their approach can be generalized to more complex permutations allowing is to use it to bound \textsc{rank-$k$}.
%
%

\subsection{Previous Work}
\label{S:related}
There is a wide literature on computing with noisy information.
Early works include results on networks with noisy gates (see \eg \cite{pippenger,BollobasB90}).
There is also a large body of work on the complexity of noisy decision trees (see \eg~\cite{pipp2,Kenyon1992}).
Reviewing this literature is beyond the scope of this paper, we thus detail the literature 
that is the closest to the topic at hand.

The seminal paper of \citet{FRPU94} initiated a long line of research on algorithms
that make noisy queries.
\citet{FRPU94} provided an algorithm to \textsc{top-$k$} with $O(n\log k)$-query complexity and $\Theta(\log n + k)$-round complexity. 
To prove the lower bound they give a specific instance and showed its hardness.
Furthermore, they showed that in the context of instances where values are binary, for any $k$ there exists such an instance, and so
 any algorithm requires at least $\Omega(n\log k)$ queries to output the correct top $k$ elements w.p. at least $2/3$.

The round complexity for these problems has also been studied. 
The earliest work on algorithms for finding the $i$'th element of
size-$n$ array
while minimizing the number of rounds and using noiseless comparisons are
due to~\cite{valiant1975parallelism,BollobasB90,4568241}. 
This problem together with the \textsc{max} and \textsc{top-$k$} has been successively studied (see~\cite{Leighton,}).
Some of these works (\eg~\cite{Gallager:2006,Feige2000}) are about computing some Boolean functions. 
In these scenarios, comparison and value queries are very similar.
Our problems share some similarity with the problem of sorting in a noisy parallel environment with concern for resampling
(see \eg~\cite{BM08,NIPS2011_4428,Makarychev:2013}).
Closely related to our problem, \citet{newman2004computing} asked whether there is a noisy decision tree for computing
the Boolean function OR using $O(n)$ queries and $O(1)$ rounds. \citet{GS10} showed that any noisy Boolean decision tree
for \textsc{max} (and so solving OR as well) using $ r \le \log^* n - O(1)$ rounds requires
$\Omega(n \underbrace{\log\log\cdots \log}_{r} n)$ queries in the worst-case scenario. They also provided an algorithm 
making $O(\log^* n)$ rounds and $O(n)$ queries. Note that our results for the \textsc{max} immediately improves
upon their bounds: our upper bound gives a trade-off between round and query complexity for any number of round and our lower
bound also applies for $r \le \log^* n$.


Recently, \citet{BMW16} 
showed that for $n$ strictly ordered (\textsc{$k$-rank}) elements the required query complexity is $\Omega(n\log n)$ queries to find the maximum w.p. at least $1-1/\poly(n)$.
Furthermore, they give an algorithm finding the $k$-Partition using $O(n \log n)$ queries and constant round complexity which is quintessentially a corollary of the noiseless case. Their work also extends to other models such as the ``erasure'' model and the noiseless model.

\medskip

There is also a variety of work on problems in either incomparable or
more general models
\cite{Rajkumar,Davidson,chen2015spectral,busa2013top,eriksson2013learning,stz2017}.
In more general models, the lower bounds do not apply to our model and the upper bounds
obtained are not competitive with ours.
We review the work on models that are the closest to ours.
Recently \citet{agarwal17c} study the Partition problem with noisy comparisons using a reduction to the coin tossing problem that replaces "Toss coin $i$"  by "Compare element $i$ to a random other element, call the result {\em Heads}  if $i$ wins". In our noisy comparison model that reduction is singularly inefficient, leading to an algorithm for Partition with query complexity larger than $n^3$ in the worst-case and lower bounds that do not
apply to our problem.
There is also a large body of related work (\eg \citet{ABM10,GGLB11,KCG16}) on best arm identification in multi-armed bandit settings, which is related yet very different: 
Rewards of multi-armed bandits are drawn according to arbitrary probability distributions and one seeks the find the machines
with the highest expected reward.
Our noisy setting is a more worst-case assumption, where w.p. $1/3$ the adversary gets to lie arbitrarily and adaptively; even if the 
adversary would always lie using the same value, the optimal solutions in both settings (value model and multi-arm best arm identification) are in general different. 
Another example is the skyline problem 
in a noisy context that has been studied recently \cite{GM15,MMV17}.
Very recently, \citet{CLM17} consider the setting of distinct values and instead 
of performing pair-wise queries one can query the maximum of a set of $\ell$ elements. 
They provide instance-optimal bounds, which in our special case of $\ell=2$ are not tight. They 
also require that the output solution is unique.

The notion of instance-optimality was introduced in the seminal work of \citet{fagin}.
Since then, it has been used to analyze popular heuristics or design
better algorithms, see for example~\cite{Afshani3046673,valiant2015instance,Baran}.
There has been recent work on nearly instance-optimal bound for best-$k$ arm identification
as well~\cite{chen17a,CLM17}.

 
%
%
%
\subsection{Preliminaries}\label{sec:prelims}
%
%

\subsection*{Notation and Definition of the Models}
The instances of size $n$ of our problems are sequences of $n$
elements that come from some multiset $V=\{v_1, v_2, \dots, v_\ell\}$ with $\ell\leq n$ values.
An  algorithm is {\em oblivious} if it does not know the multiset $V$ ahead of time.
All logarithms are to the base of $2$ unless stated otherwise.

\begin{notation}
  Let $ \log_b^*(n)$ be the least integer $i$ such that $a_i \le 0$
  in the sequence defined by $a_0=n$ and $a_{i+1}=\log_b a_i$.
\end{notation}

\begin{notation}[Tower functions] Let $b>1$.
Let $b\uparrow\uparrow i$ be the number $b_i$ in the sequence defined by $b_1=b$ and $b_{j+1}=b^{b_j}$.
 Let $\zeta^{b,\delta}_i$ be $(b/\delta) \uparrow \uparrow i$.
It's worth pointing out that $b\uparrow\uparrow i$ is the inverse function of $\log^*_b(i)$. 
\end{notation}

\begin{definition}
  We define the \emph{noisy value model} as follows.
  Given a set $S$ of $n$ elements, the algorithm has access to elements of $S$ via a query oracle. 
  To answer a query about element $i$,  the 
  oracle, with probability $2/3$, returns the true value of $i$,
  and with the remaining probability, returns an arbitrary value.
\end{definition}

\begin{definition}
  We define the \emph{noisy comparison model} as follows.
  Given a set $S$ of $n$ elements, the algorithm 
  has access to elements of $S$ via a query oracle.
  The query oracle only answers
  queries of the following form
  ``is the value of element $x$ greater or equal than
  the value of element $y$?''.
  For a given query, the 
  oracle, with probability $2/3$, returns the correct answer
  to the query,
  and with the remaining probability, returns an arbitrary answer.  
\end{definition}

The following lemma states that any algorithm with constant success probability
for the noisy comparison model can be emulated by an algorithm in the noisy value model by only losing a constant factor in
the query complexity. Using this reduction, all our lower bounds for the noisy
value model apply to the noisy comparison model (up to a constant factor in the query complexity).
The proof can be found in \autoref{sec:relationship}. 
\begin{lemma}[From value to comparison queries]\label{lemma:reduction}
     If there exists an algorithm $A$ solving a problem $P$ in the noisy comparison model with
     query complexity $q$ and round complexity $r$ with correctness probability at least $1-\delta$,
     then there exists an algorithm $B$ solving $P$ in the noisy value model with query complexity 
     $9q$, round complexity $r$, and correctness probability at least $1-\delta$. 
\end{lemma}

Therefore, in the
rest of the paper we will focus on lower bounds for the noisy value model. In the same spirit, all our upper bounds are in the noisy comparison model\footnote{With exception of the \textsc{threshold}-$v$ problem  which is not defined in the comparison model.}. 
As mentioned earlier, our results imply the following.
Generally speaking there is no reduction in the other direction---from the noisy comparison model to the
noisy value model--- without losing a super-constant factor in the query complexity. Nonetheless, in many settings of interest we prove, by means of proving matching bound, that the query complexity is up to constants the same.


\subsection*{Instance-Optimality}
Instance-optimality was introduced by Fagin, Lotem and Naor in their 
seminal paper~\cite{fagin}. 
They originally developed this notion a remedy to worst-case analysis
based on the empirical observation that for several problems most
practical instances can be solved efficiently
and very few artificial instances that are computationally `hard'.
Therefore, we would wish that an algorithm is able to recognize whenever 
it is given a considerably easier input.
This is formalized in the concept of instance-optimality:
An algorithm is called instance-optimal if for all instances $I$, it 
is asymptotically as  `efficient' (\eg, query complexity, runtime, \etc) 
as the most efficient algorithm for instance $I$.
Of course, if the algorithm knows $I$ in advance, then it can output the 
correct solution immediately. Thus, we focus on algorithms that
know the instance up to a permutation of the input elements. 

Thus, an instance-optimal algorithm is asymptotically the best possible 
algorithm one could hope for (w.r.t. to the measure of interest). The 
notion has been widely-used to provide a fine-grained 
analysis of various algorithms (see, \eg \cite{Afshani3046673}).

We now sate the concept more formally.
 Let $\mathcal{I}$
be a class of inputs for a given problem.
 Let $\mathcal{A}$ denote the set of algorithms that are, on every input, correct with probability at least $2/3$. 
 \begin{definition}\label{def:ionew}
 An algorithm $B$ is \emph{instance-optimal} if for every input $J \in \mathcal{I}$, the output is correct with probability at least $2/3$, and 
the complexity is 
    \[ O\left(\inf_{A \in \mathcal{A}} \sup_{\text{$I$ permutation of $J$}} (\text{expected complexity  of $A$ on $I$})\right). \]
 \end{definition}
 
 Thus, for instance-optimality, we compare $B$ to an algorithm $A$ that is allowed to have arbitrarily high query complexity on inputs that are not permutations of $J$: that does not affect the above quantity. 
 
 For comparison with the usual notion of optimality, note that $B$ is worst-case optimal if for every input $J$, the output is correct with probability at least $2/3$ and the complexity is 
                      \[ O(\inf_{A \in \mathcal{A}} \sup_{I\in \mathcal{I}} (\text{expected complexity  of $A$ on $I$})). \]

 \begin{remark}
Our definition of instance-optimality is in line with the original definition from  \cite{fagin} who write: ``the cost of an instance-optimal algorithm is essentially the cost of the shortest proof''. In contrast, the recently more popular variant of \cite{Afshani3046673} allows $A$ to be incorrect on inputs that are not permutations of $J$: 
                     \[O\left(\inf_{\substack{\text{$A$ correct w.p. $2/3$} \\ \text{for permutations of $J$}}} \sup_{\text{$I$ permutation of $J$}} (\text{expected complexity of $A$ on $I$})\right). \]
Thus, for example for the problem of finding the maximum, in the case of an input $J$ whose elements are all equal,\cite{Afshani3046673}] compares $B$ to an algorithm with query complexity $0$ on $J$, whereas for us (and for \cite{fagin})  algorithm $A$ still needs to query for verification purposes, since it is required to be probably correct on all inputs.  

 \end{remark}

\begin{lemma}\label{lem:Omegan}
Any algorithm $A$ for any of the following problems \textsc{max}, \textsc{Threshold}, \textsc{Top-$k$}, or \textsc{Rank}-$k$,
that is correct with probability at least $2/3$ on all instances, is such that, $T(A,I)=\Omega(n)$ for any instance $I$. 
\end{lemma}

\begin{proof}
Even if the queries are error-free, and even if the adversary provides an input $I$ to $A$ with the promise that the correct input either is $I$ or is $I$ with one element changed so that its modified value is larger than all other values, in order to be correct with probability $2/3$ algorithm $A$ must still query at least $\Omega(n)$ elements.
\end{proof}

\subsection*{Expected vs Worst-Case Number of Queries}
In all our lower bounds, we consider algorithms that have constant success probability and make
a deterministic number of queries. We can use the following lemma
to move from a lower bound on the worst-case number of queries to a lower bound
on the expected number of queries an algorithm must perform.

\begin{lemma}
\label{lem:wcvsexpct}
  For any problem $P$, if any algorithm that solves $P$ with probability at least 2/3 and
  must make at least $q$ queries then any algorithm that solves $P$ with probability
  at least 5/6 must make at least $q/6$ queries in \emph{expectation}.
\end{lemma}
\begin{proof}
  We simply use Markov's inequality. Suppose that there exists an algorithm that
  solves $P$ with probability at least 5/6 and that makes less than $q/6$ queries
  in expectation.
  Then, by Markov's inequality, the probability that the algorithm makes more than
  $q$ queries is at most $1/6$. Consider running this algorithm and in
  each execution that makes more than $q-1$ queries, stop at the $q-1$st query and output a random
  solution. The success probability of this algorithm is now at least 5/6 - 1/6 = 2/3, and the worst-case
  number of query less than $q$, a contradiction.
\end{proof}

In all our lower bounds, the success probability can be changed to an arbitrary constant.
of queries performed by the algorithm.

\subsection*{Organization of the Paper}

In \autoref{sec:query}, we give tight bounds for the worst-case complexity of
\textsc{max},\textsc{top}-$k$,\textsc{threshold}-$v$.
We start with \textsc{max}, the problem of finding the maximum, which is a building brick in many of our algorithms
and a warm-up that illustrates some of the divide and conquer ideas we use
in the paper (\autoref{sec:max}). In addition, we present tight trade-offs for the query vs the round complexity.

In  \autoref{sec:instanceopt}, we study the instance-optimal query complexity of 
the  \textsc{threshold}-$v$ problem  and \textsc{top}-$k$.
%
The section also contains our approximation algorithm: an 
algorithm returning $k$ elements among the top-$(1+\eps)k$ 
elements.

\autoref{sec:instanceopt} presents our instance-optimal upper and lower bounds for the query complexity of
\textsc{max}, \textsc{threshold}-$v$, and \textsc{top}-$k$.

Finally, in \autoref{sec:rank}, we show a separation between the value- and the comparison model. The section dissects the query complexity of \textsc{rank$-k$} in the value model and in the comparison model. 




%


\section{Tight Query Complexity}\label{sec:query}

We start with the \textsc{max} problem (\autoref{sec:max}), followed by 
\textsc{threshold}-$v$ (\autoref{sec:threshold}) and 
\textsc{top}-$k$ (\autoref{sec:topk}).

%
%
%
\subsection{The \textsc{max} Problem (with tight round complexity) }\label{sec:max}

Consider the problem of computing the maximum with noisy comparison
or value queries, 
where we are interested in three parameters: the correctness
probability, the number of queries, 
and the number of rounds.
Note that
when several input elements share the same value, there may be several
correct outputs for $\Max$. 

Our upper bound generalizes \citet{GS10}, which show 
that given $n$ distinct elements, using $O(n)$ queries, $\log^*(n)$ 
rounds are sufficient and necessary to compute the max.
We generalize this to an arbitrary trade-off between
query and round complexity. In particular, 
for $r$ rounds, the query complexity is at  most
$q\leq n\underbrace{\log\log\cdots \log}_{r-O(1)} n$.
Our algorithm works for both comparison and value models.
As for the lower bound (\autoref{theorem:maxLB}), we generalize the
result of~\citet{GS10}, by giving showing that the trade-off obtained
by our algorithm is optimal, up to an additive constant term.
Our lower bound applies
to both value and comparison models. 

\begin{theorem}\label{lem:topksidekicks} {\bf (Algorithm for Max)}
Consider the $\textsc{Max}$ problem in the noisy comparison model.
  Fix a set of elements $X$, an integer $r$, and $\delta \in [1/n^8, 1]$.
  Algorithm $\Max(X,r,\delta)$  returns the maximum element of $X$  with
  correctness probability at least $1-\delta$ and
  has 
  \begin{enumerate}
    \item round complexity $r$ and 
  \item query complexity $q=O(n\log (b/\delta))$,  where $b$ is defined
  by $r=\log_b^*(n) +4$.
  \end{enumerate}
\end{theorem}

Two special cases of \autoref{lem:topksidekicks} that are of interest:
\begin{itemize}
\item
  for $\delta=1/3$ and $r=\log^*(n) +O(1)$ rounds, the query
  complexity is $q=O(n)$, which is the minimum possible number
  of queries; and
\item
  for $\delta=1/3$ and $r=O(1)$ rounds, which is the minimum
  possible number of rounds, the query complexity is
  $q=n\underbrace{\log\log\cdots \log}_{O(1)} n$.
\end{itemize}

This result has to be contrasted with the following lower bound.


\begin{theorem}\label{theorem:maxLB} {\bf (Worst-case Lower bound for Max)}
Consider the $\textsc{Max}$ problem in the noisy value model.
There exists a set of elements $X$ such that for any 
algorithm $A$, for any $b\geq 4$.
the following holds. 
Suppose $A$  returns the maximum of $X$ with correctness probability at least
$2/3$ in $r$ rounds 
using at most $q=n \log_3 (b)$ value queries, then it must be that $r\geq \log_b^*(n)-O(1)$.
  
\end{theorem}


Together, the above two theorems provide a complexity
characterization of the trade-off between 
query complexity and round complexity for the noisy maximum problem.



\subsubsection{Upper Bound (comparison model) - Proof of \autoref{lem:topksidekicks} }

\vincent{I am not sure what to do with the following paragraph. Is that really important? it seems
  a bit off as it is because we are saying that we assume an error probability of 1/8 but in the prelims
  we say that the error is 1/3 throughout the paper, maybe we can just get rid of this paragraph?}
To handle the possibility of ties between elements, to compare a pair $\{ x,y\}$
the algorithm systematically queries both ``$x\geq y$?" and ``$x\leq y$?", repeating those
two queries until it gets two answers that are consistent with each other; in this way, the comparison has three possible outputs, $x>y$, $x=y$ and $x<y$. If each query has error probability at most $1/8$, then the output of the comparison is correct with probability at least $2/3$.

First we analyze the problem when there is a single round. The one-round algorithm will be a subroutine of our main, multi-round algorithm.
\begin{algorithm*}\caption*{$\operatorname{\bf Algorithm}$\ $\OneRoundMax(X,\delta)$  \  \ \ \ \ (see \autoref{lemma:oneround})
\\
{\bf input:}  set $X$, error probability $\delta$ \\
{\bf output:} largest element in $X$\\
{\bf error probability:} $\delta$
}
\begin{algorithmic}
\STATE let $c=48$
\STATE  \inparallel, compare each pair $\{ x,y\}$ in $X$ to each other exactly $c\log(|X|/\delta)+1$ times
\IF {there exists an element $x$ such that for every $y$, at least half of the comparisons between $x$ and $y$ lead to $x>y$ or to $x=y$}
\STATE \Output $x$
\ELSE 
\STATE \Output FAIL
\ENDIF
\end{algorithmic}
\end{algorithm*}

\begin{lemma}\label{lemma:oneround}
Algorithm $\OneRoundMax$ has query complexity $O(|X|^2\log(|X|/\delta))$, round complexity $1$, and error probability at most $\delta$.
\end{lemma}

\begin{proof}
The query complexity and round complexity statements are obvious. 

To analyze the probability of error, we will prove that some maximum element (probably) passes the algorithm's test, and no non-maximum element does.  Let $x$ be a maximum (since ties are allowed, there could be several maximum elements that are all equal). Fix some element $y$ and let $Y_x(y)$ denote the random variable equal to the number of comparisons between $x$ and $y$ leading to either $x>y$ or $x=y$. We use multiplicative Chernoff bounds (\autoref{lemma:Chernoff}) to bound the probability that $Y_x(y)$ is at least $3/4$ of its expectation, and taking Union bound, we conclude that with probability at least $1-\delta/2$ there is at least one maximum element $x$ such that  for every $y$, at least half of the comparisons between $x$ and $y$ lead to $x>y$ or to $x=y$. 
Indeed, since each comparison is correct with probability at least $2/3$, we have $\E{Y}\geq (2/3)c\log (|X|/\delta)$. So
Chernoff bound bounds bound the probability of failure  $e^{-(1/16)(2/3)c\log(|X|/\delta)/2}={\delta}/(2|X|)$ since $c=48$.
Thus, among all maximum elements, at least one (probably) passes the algorithm's test.

Now, consider an element $x'$ that is not a maximum. Then the probability that at least half of the comparisons between $x'$ and the true maximum $x$ lead to $x'>x$ or to $x'=x$  is at most $\delta/(2|X|)$. By the Union bound, with probability at least $1-\delta/2$ there exists no non-maximum element $x'$ such that at least half of the comparisons between $x'$ and the true maximum $x$ lead to $x'>x$ or to $x'=x$.

By the Union bound, with probability at least $1-\delta$ both events hold, and then the output of the algorithm is correct.
\end{proof}

\FloatBarrier
We now state the multi-round algorithm $\Max$ that builds on $\OneRoundMax$. 
The main idea is to first sample a set $Y$ of expected size $n^{2/3}$, where $n$ is the size of the input.
In $Y$ we sample a subset $Z$ of size $m^{1/3}$. The small size allows us to do all pairwise comparisons in $Z$ to find its maximum element $z_1$ whose approximate rank is $n^{2/3}$. 
We then compare every element in $Y$ to $z_1$. Note that the maximum element in $Y$ has approximate rank $n^{1/3}$. 
We can thus compute in linear time the maximum element $y_1$ of $Y$.
The crucial property is that its approximate rank is $n^{1/3}$. This allows us to find all elements in the input $X$ that  
are larger. 
The difficulty is that we seek accomplish this in query complexity $O(n \log(b/\delta)$.
To do so we proceed in rounds in which the number of queries per remaining candidate scales as $\zeta_t$ (tower function).
We rely on three properties 1) are only approximately $n^{1/3}$ many in elements $X$ that are larger than $y_1$ and 
2) every time we eliminate an element we have additional query complexity in the next round.
3) The largest element of $X$ will never get eliminated.

Since there are only approximately $n^{1/3}$ many in elements $X$ that are larger than $y_1$, after all these rounds there will
be only approximately $n^{1/3}$ elements left. We can simply compare them pairwise to find the maximum value of $X$.

\begin{algorithm}\caption*{$\operatorname{\bf Algorithm}$\ $\Max(X,r,\delta)$ \  \ \ \ \ (see \autoref{lem:topksidekicks})
\\
{\bf input:}  set $X$ of size $n$, number of rounds $r$, error probability $\delta\in[1/n^{8},1]$\\
{\bf output:} largest element in $X$\\
{\bf error probability:} $\delta$
}
\begin{algorithmic}
  \STATE let $b$ be such that $r = \log^*_b(n) + 4$
  \STATE let $\zeta_i$ be $(b/\delta) \uparrow \uparrow i$
      \STATE  $Y \gets $ random sample of $X$, each element of $X$ being taken with probability $n^{-1/3}$
  \STATE $Z \gets $ random sample of $Y$, each element of $Y$ being taken with probability $n^{-1/3}$
  \STATE  $z_1\gets \OneRoundMax(Z,\delta/5)$   
  \STATE $c\gets 24$
  \STATE \inparallel, compare each element in $Y$ to $z_1$ exactly $c\ln(n/\delta)+1$ times  
  \STATE  $Y^* \gets \{ z_1\} \cup \{$  elements of $Y$ that are assessed to be strictly greater than $z_1$ for at least half of their comparisons to  $z_1\}$
  \STATE $y_1\gets\OneRoundMax(Y^*,\delta/5)$
  \STATE $X_0 \gets X$
  \FOR {$t \in [1, r-4]$} 
  \STATE $n_t\gets t \cdot  144 \frac{n}{2^t|X_{t-1}|}\ln(16b/\delta)$
    \STATE \inparallel, compare each element in $X_{t-1}$ to $y_1$  exactly $n_t$ times  
  \STATE  $X_t  \gets\{$elements of $X_{t-1}$ that are assessed to be strictly greater than $y_1$ for at least half of their comparisons to  $y_1\}$
  \STATE  If $|X_t| > \max\left\{ \frac{n}{ 2^t \zeta_t  }, 2n^{1/3}\ln(16/\delta)^2\right\}$ then \Output FAIL
  \ENDFOR

\STATE  \Output $\OneRoundMax(X_{r-4} \cup \{y_1\}, \delta/5)$
\end{algorithmic}
\end{algorithm}

\FloatBarrier
The proof of \autoref{lem:topksidekicks} relies on    \autoref{obs:somethingabouty1}, \autoref{lem:invarlemma} and \autoref{lemma:correctness}.
\begin{lemma}\label{obs:somethingabouty1}
Consider the element $y_1$  computed in $\Max$.
Let $\mathcal{E}$ be the event that the rank of $y_1$  in $X$  is at most $n^{1/3} \ln(5/\delta)$ (the randomness is over the computation of $y_1$).
Then,
\[ \Pr{\mathcal{E}}\geq 1- (2/5+1/n)\delta. \]
%
\end{lemma}

\begin{proof}
We first prove that the maximum element of $Y$ is (probably) in $Y^*$. If $z_1$ is maximum, that is obvious. Else, let $y^*$ be a maximum element  of $Y$, strictly greater than  $z_1$; we proceed similarly to the proof of \autoref{lemma:oneround}. Let $W$ denote the random variable equal to the number of wins of the $y^*$ when compared to $z_1$. Since each comparison is won by $y^*$ with probability at least $2/3$, the expectation is $\E{W}\geq (2/3)c\log (n/\delta)$. Consider the probability that $Y$ is less than $(c/2)\log(n/\delta)$, that is, is at most $2/3$rds of its expectation. By multiplicative Chernoff bounds (\autoref{lemma:Chernoff}) that is at most 
\[e^{-(1/9)(3/4)c\log(n/\delta)/2}=\frac{\delta}{n}.\]
 Thus $Y^*$ contains the maximum element of $Y$ with probability at least $1-\delta/n$.

Conditioning on this, the probability that $y_1\neq y^*$  is then at least  $1-\delta/5$ (due to \autoref{lemma:oneround}). 
If we consider the elements of $X$ by decreasing order,  since each is placed in $Y$ with probability $n^{-1/3}$,  the probability that more than $n^{1/3}\ln(5/\delta)$ elements are considered before one of them ($y^*$) is placed into $Y$ is 
\[(1-n^{-1/3})^{n^{1/3}\ln(5/\delta)}\leq (1/e)^{\ln(5/\delta)}=\delta/5.\]
\end{proof}

The loop of the algorithm rapidly weeds out of $X$ the elements that are less than or equal to $y_1$, until there are few enough candidate maxima that algorithm $\OneRoundMax$ can be run efficiently, taking care at the same time to not accidentally eliminate the maximum element of $X$. 
%

\begin{lemma}\label{lem:invarlemma}
Condition on Event $\mathcal{E}$.
Consider a large enough $n$.  Let $S_t=X_t\cap \{  s\colon s\leq y_1  \}$.
Let Event $\mathcal{E}'$ be that for any $t \leq \log_b^*(n)$ we have:
$ |S_{t}| \leq \max\left\{ \frac{n}{ 2^t \zeta_t  }, 2n^{1/3}\ln(16/\delta)^2\right\}$.
Then,
\[ \Pr{\mathcal{E}'~|~\mathcal{E}}\geq 1- \frac{\delta t}{n^9} - \frac{\delta}{16}\sum_{j=1}^{t} \frac{1}{2^{j}}  . \]
\end{lemma}

\begin{proof}
  The claim holds trivially for $t=0$, due to or conditioning on $\mathcal{E}$.
  Assume the claim holds for $t-1$ and that $|S_{t-1}| \geq 2n^{1/3}\ln(16/\delta)^2$. Note that by the inductive hypothesis,
  this happens with probability at least
  $1-\delta\frac{t-1}{n^9}- \frac{\delta}{16}\sum_{j=1}^{t-1} \frac{1}{2^{j}}   $.
  Note that  $|X_{t-1}| \leq |S_{t-1}| + n^{1/3} \ln(16b/\delta)^2 \leq 2|S_{t-1}| $.
  
  Let $s\in S_{t-1}$ and $T$ be the random variable denoting the number of outcomes $s \leq y_1$ when comparing $s$ to $y_1$ in iteration $t$. 
  
  If $c_t$ denotes the number of comparisons of $s$ with $y_1$, then $\E{T}$ is at least $(2/3)c_t$, and by multiplicative Chernoff bounds (\autoref{lemma:Chernoff}) the probability that $T < (1/2)c_t$ is upper bounded by $e^{- (3/4)^2c_t/3}$. 
  The number of comparisons of $s$ to $y_1$ is at least:
  \[\frac{144n \ln(16b/\delta)}{2^t |X_{t-1}|} \geq  \frac{144  n\ln(16b/\delta)}{2^t 2|S_{t-1}|} \geq   
    \frac{144n  \ln(16b/\delta)}{2^t 2\frac{n}{ 2^{t-1} \zeta_{t-1} }} 
\geq 
36   \zeta_{t-1}\ln(16b/\delta).\]
Thus, we have 
$\E{T}\geq (2/3)36 \zeta_{t-1}  \ln(16b/\delta)$. The probability that  $s\in S_t$, that is, that
at least half of the comparisons to $y_1$ (erroneously) returned $s$ to be the winner, is  the probability that $T$ exceeds its expectation by a factor of at least $3/2$. By multiplicative Chernoff bounds,
\begin{align}\label{lampe}
  \Pr{ s\in S_t \mid s\in S_{t-1}} &\leq \Pr{T\geq (3/4)\E{T}} \nonumber 
                                   \leq \exp\left(- (3/4)^2  \frac{(2/3)36 \zeta_{t-1}  \ln(16b/\delta)}{3} \right) \nonumber \\
                                   &= \frac{1}{\left(\frac{16b}{\delta}\right)^{t \cdot  \zeta_{t-1}}}
                                   \leq \frac{1}{2\cdot 2^t  \zeta_t},
\end{align} 
where we used that $\zeta_t = (b/\zeta)^{2\zeta_{t-1}}$. 
Now, $|S_t|$ is a sum of $|S_{t-1}|$ independent random 0/1 variables, and $\E{|S_t|}\leq |S_{t-1}|/ ( 2\cdot 2^t   \zeta_t)$.
We use multiplicative Chernoff bounds again. There are two cases. First, if $2\E{|S_t|}\geq n^{1/3} \log(16/\delta)^2$, then
\begin{align*}
  \Pr{|S_t|>\frac{n}{2^t   \zeta_t}} \leq   \Pr{|S_t|> \frac{ |S_{t-1}|}{2^t   \zeta_t}   }&\leq \Pr{|S_t|>2\E{|S_t|}}\\
                                                                                       &\leq e^{-  \E{|S_{t}| }/3}\leq
                                                                                         e^{-  n^{1/3}\log(16/\delta)^2/3} < \delta/(2 n^{9}).
\end{align*}

Second, if $n^{1/3}\log(16/\delta)^2> 2\E{|S_t|}$, then let $\alpha = \frac{n^{1/3}\log(16/\delta)^2}{2\E{|S_t|}}$, we have
\begin{align*}
  \Pr{|S_t| > n^{1/3}} = \Pr{|S_t| > \alpha \cdot 2\E{|S_t|}} < e^{-(\alpha-1)^2 2\E{|S_t|}/3} \le e^{-\Theta(n^{1/3}\log(16/\delta)^2)}< \delta/(2n^9),
\end{align*}
where the penultimate inequality is implied by $n^{1/3} \log(16/\delta)^2> 2 \E{|S_t|}$.
By Union bound, the lemma follows.
\end{proof}

The following lemma helps to bound the error probability.

\begin{lemma}\label{lemma:correctness}
  Condition on Event $\mathcal{E}'$.
  Consider a large enough $n$.  Let $x_1$ denote the true maximum.
  For any $1 \le t \leq \log_b^*(n)$, with probability at least $1-\delta/2$ we have $x_1 \in X_t \cup \{y_1\}.$
\end{lemma}

\begin{proof}
Assume $x_1\in X_{t-1} \cup \{y_1\}$. In order to have $x_1 \not\in X_t \cup \{y_1\}$ we have that more than half of the comparisons of $x_1$ with $ y_1$ must (erroneously) return that $y_1$ wins. Let $U$ be the random variable denoting the number of comparisons between $x_1$ and $y_1$ in iteration $t$ that are won by $y_1$. 

Since there are $36    \zeta_{t-1}\ln(16b/\delta)$ comparisons of $x_1$ with $y_1$ in that iteration, we have 
$\E{U}=(1/3)36    \zeta_{t-1}\ln(16b/\delta)$. The probability that  $x_1\notin X_t$, that is, that
 at least half of the comparisons to $y_1$ (erroneously) returned $y_1$ as the winner, is  the probability that $U$ exceeds its expectation by a factor of at least $3/2$. 
$U$ has the same distribution as $T$ studied in the proof of \autoref{lem:invarlemma}, so by~(\eqref{lampe}).

We conclude $\Pr{ x_1\notin X_t \mid x_1\in X_{t-1}} \leq \frac{\delta}{2^t\cdot   \zeta_t}$.
Summing over $i=1,2,\ldots, \log_b^*(n)$, $\Pr{x_1\notin X_t}\leq \delta/2$, as desired.
\end{proof}

\begin{proof}[Proof of \autoref{lem:topksidekicks}]
 
 Consider $\Max$. First we study the probability that the output is correct. 
  We  condition on Event $\mathcal{E}'$. By \autoref{lem:invarlemma}, the algorithm does not output FAIL.
  We have by \autoref{lemma:correctness} that $x_1 \in X_{r-4} \cup \{y_1\}$ with probability
  at least $1-\delta/2$. Now, we have that by \autoref{lemma:oneround} that $\OneRoundMax(X_{r-4} \cup \{y_1\}, \delta/n^9)$
  yields the maximum of $X_{r-4}\cup\{y_1\}$ with probability at least $1-\delta/n^9$. Hence taking a Union bound
  over the failure probability of Events $\mathcal{E}$ (\autoref{obs:somethingabouty1}) and $\mathcal{E}'$  and of $\OneRoundMax(X_{r-4} \cup \{y_1\}, \delta/n^9)$, we have that the
  probability of success is at least $1-\delta$, as desired.
  	
  The round complexity can be bounded as follows. With exception of the for-loop, which takes $r -  4 = \log_b^*(n)$ rounds,
  there are three calls to $\OneRoundMax$, which takes a total of three rounds, and the computation of $Y^*$ which takes one round.

  The query complexity can be bounded as follows.
  First, observe that the total number of queries before the for-loop is at most $O(n\log(1/\delta))$.
  Now, consider the outer for-loop. We have a total of at most $144 n \log(16b/\delta) \sum_{t=1}^{\log_b^*(n)}  2^{-t} = O(n \log (b/\delta))$
  queries. Finally, since we condition on Event $\mathcal{E}'$, we have by \autoref{lem:invarlemma} that
  $|X_{r-4}| \le \max{2n^{1/3}\ln(16b/\delta)^2, n/(2^t  \zeta_t)}$.
  Now, by definition of $  \zeta_t$, we have that $n/(2^t  \zeta_t) \le 2n^{1/3}\ln(16b/\delta)^2$ 
  and therefore, 
  $|X_{r-4}| \le 2n^{1/3}\ln(16b/\delta)^2$. Hence, the last call to $\OneRoundMax$ generates
  at most $O(n)$ queries (by definition of $\delta$).
\end{proof}

\subsubsection{Lower bound (value model) - Proof of \autoref{theorem:maxLB}} 

In order to prove {\autoref{theorem:maxLB}} we will use the following technical lemma.
\begin{lemma}\label{lem:helper} Consider the sequence $(x_t)_{t\geq 0}$ defined by: 
\[x_t =
\begin{cases}
1 & t=0 \\
4x_{t-1}b^{ 5^t x_{t-1}} & t\geq 1
\end{cases} .
\]
Then $x_t \leq
 \frac{ (10b)\uparrow\uparrow (t+1) }{10^{t+1}}$.
\end{lemma}
\begin{proof}
%
The proof is by induction on $t$. 
%
%
For $t=0$ the claim holds. Now consider general $t\geq 1$.  We have for $t\geq 1$ that $4x_{t-1}\leq b^{ 5^t x_{t-1}} $. Hence, 
\begin{align*} x_t &\leq  b^{ 2\cdot 5^t x_{t-1}} \leq
b^{ 2\cdot 5^t   \frac{(10b)\uparrow\uparrow t}{10^t}  }
\leq b^{ (10b)\uparrow\uparrow (t+1) }
\leq \frac{ 10^{ (10b)\uparrow\uparrow (t+1) }  } {10^{t+1}} b^{ (10b)\uparrow\uparrow (t+1) }
= \frac{ (10b)\uparrow\uparrow (t+1)} {10^{t+1}}
.
\end{align*}
\end{proof}

The following lemma implies \autoref{theorem:maxLB}.
\begin{lemma}\label{lemma:maxLBrankmodel}
    Let $b\geq 4$. Consider an algorithm $A$ for finding the unique maximum in the noisy value model with using  at most $q=n \log_{3} (b)$ queries and  at most $\log_b^*(n)-c$ rounds, for some large enough constant $c$. Then $A$ returns the wrong maximum w.p. at least $1/3$.
   \end{lemma}
\begin{proof} 
The adversary strategy is as follows: when queried for the value of some element $x$,  it answers with the true value of $x$ with probability $2/3$, and with the complementary probability, $1/3$, the adversary answers with the value of rank $1$. In particular, the adversary always responds with rank $1$ when when the rank of the maximum element $x_1$ is queried.

In the following we assume that the algorithm knows the strategy of the adversary. 
For all rounds $t \geq 1$, we define $S_{t-1}$ to be the set of elements whose queries have always
returned rank $1$ and that have received exactly as many queries as $x_1$ during the first  $t-1$ rounds.

Partition the elements of $S_{t-1}$ into two classes: 
\begin{itemize}
\item Class 1 consists of those elements on which the algorithm  spends strictly more than $q'=5^tq/|S_{t-1}|= \log_3(b^{5^t n/|S_{t-1}| } )$ queries
during round $t$, and 
\item Class 2 consists of those elements on which the algorithm spends at most $q'$ queries.
Since the algorithm is limited to $q$ queries, at most $|S_{t-1}|/5^t$ elements of $S_{t-1}$ are in Class 1. 
\end{itemize}	
For the sake of proving our lower bound, we use the ``little birdie principle'' and provide the algorithm with additional information:
at the end of round $t$, the adversary reveals the true rank of all elements in Class 1, and answers for free additional (noisy)
queries about elements in Class 2, so that every element in Class 2 has exactly $q'$ queries during round $t$. This can only help the algorithm. 
	
		For the algorithm, at the beginning of round $t$ the elements of $S_{t-1}$ are indistinguishable from one another and from the unique maximum $x_1$.
		Hence, assuming a random permutation of the input, 
		we have that the probability that $x_1$ is in Class 1 is at most $5^{-t}$. Should $x_1$ be in Class 1, then without loss of generality the algorithm `wins' immediately and no further queries are required. 
		 At the end of round $t$, $S_t$ consists of those elements of $S_{t-1}$ that were in Class 2 and whose $q'$ queries were all answered with rank $1$. 	
		
		Let  $t'= \arg\max_t \{ b\uparrow\uparrow (t+2) \leq \sqrt{n}\}$. Note that $t'\geq \log_b^*(n)-O(1)$
In the following we show, by induction, that for all $t \leq t' $, w.p. at least $1-\sum_{1\leq \tau \leq t} (5^{-\tau}+\frac{1}{n^9})$ it holds that
\begin{align}\label{eq:loco} 
|S_t| \geq  n/x_t, 
\end{align}
where $x_t$ is defined in \autoref{lem:helper}. For $t=0$ the claim holds trivially since $|S_0|=n$ and $x_0=1$. 
Consider $t\geq 1$ and condition on $|S_{t-1}|\geq n/x_{t-1}$. We have 
\begin{align*}		
		\E{|S_{t}|~\mid~\mathcal{F}_{t-1}} &=  
\frac{|\{ x \colon x\in S_{t-1}\cap  \text{Class 2} \} |}{3^{q'}} \geq
\frac{|S_{t-1}| -  |\{ u \in \text{Class 2} \} |}{3^{q'}} \geq 
\frac{|S_{t-1}| (1-5^{-t})}{3^{  5^tq/|S_{t-1}|  }} \\
&\geq 
\frac{|S_{t-1}| }{2\cdot b^{  5^t  n    /|S_{t-1}|  }} \geq   \frac{n  }{x_{t-1} 2\cdot b^{  5^t  x_{t-1}  }} = 2\frac{n}{x_t} ,
 \end{align*}
 where $\mathcal{F}_t$ denotes the filtration up to time $t$.
 \autoref{lem:helper} implies that 
\[ x_t\leq \frac{ (10b)\uparrow\uparrow (t+1) }{10^{t+1}} \leq (10b)\uparrow\uparrow (t+1)  \leq b\uparrow\uparrow (t+2) 
\leq b\uparrow\uparrow (t'+2) \leq \sqrt{n}  ,\]  for all $b\geq 4,t\geq 1$.  
Thus,   $\E{|S_{t}|~\mid~\mathcal{F}_{t-1}}  \geq \frac{n}{x_t} \geq  \sqrt{n}$.

Therefore,  multiplicative Chernoff bounds imply that with high probability $|S_t|\geq (1/2)\E{|S_{t-1}|}$:
$$  \Pr{ |S_t| \geq \frac{\E{|S_{t}|}}{2}} = 1- \Pr{ |S_t| < \frac{\E{|S_{t}|}}{2}}  \geq  1- \exp\left(- \left(\frac12\right)^2 \frac{\E{|S_{t-1}|}}{2} \right) \geq 1-1/n^9.  $$
Thus,
\[ |S_{t}|  \geq  \frac{\E{|S_{t}|}}{2} \geq  \frac{n}{x_t},  \]
w.p. at least $ 1-\sum_{1\leq \tau \leq t-1} (5^{-\tau}+\frac{1}{n^9})  - 5^t - \frac{1}{n^9} =  1-\sum_{1\leq \tau \leq t} (5^{-\tau}+\frac{1}{n^9}) $, where we used Union bound.
This concludes the inductive step and proves \eqref{eq:loco}.

Suppose the algorithm has not found $x_1$ after $t'$ rounds and outputs an element of $S_{t'}$ chosen u.a.r., then its success probability (given that $x_1\in S_{t'}$) is at most 
\[ \frac{1}{|S_{t'}|} \leq \frac{x_{t'}}{n} = \frac{1}{\sqrt{n}}. \]
Thus, with query complexity bounded by $q=n\log_3(b)$ and round complexity bounded by $r=\log^*_b(n)-\Theta(1)$, the output is correct with probability at most 
$\sum_{1\leq t\leq r} (5^{-t}+\frac{1}{n^9})
+\frac{1}{\sqrt{n}}<2/3$.
\end{proof}

%
%
%
\subsection{The \textsc{Threshold}-$v$ Problem}\label{sec:threshold}
%
%

In this section, we focus on the Threshold-$v$ problem: given a
multiset $V$ of elements, find all the elements whose values are
at least $v$. Note that this problem is only defined in the
noisy-value-model.
We start by giving an upper bound for the threshold-$v$ problem.
 
\begin{theorem}{\bf(Oblivious Algorithm
    for Threshold-$v$)} \label{pro:mightyUB}
Consider the $\textsc{Threshold-v}$ problem in the noisy value model. 
Fix a set of elements $X$ with values $V(X)=(v_1,v_2, \dots, v_n)$, and a value $v$.
The oblivious algorithm $\Threshold$
returns the elements of $X$ with value at least $v$ with correctness probability at least $2/3$
and has
\begin{enumerate}
\item expected query complexity $\E{q}=O\left( n + n \log (\min \{k_v,n-k_v\} ) \right),$ 
where $k_v$ is the number of such elements and 
\item expected
round complexity $\E{r} = O(\log_2^* n\cdot \log \log k_v).$
\end{enumerate}
\end{theorem}
Note that a worse-case run time of $O\left( n \log (\min \{k_v,n-k_v\} + 1) \right)$ cannot be achieved since the algorithm does not know $k_v$ and has to estimate it.

We then show that the worst-case
query complexity cannot be improved, no matter
the number of rounds allowed. Motivated by
rating and grading systems where the number of distinct values is
bounded, we turn to instances where the number of distinct values
is bounded by $n^{1-\epsilon}$ for some constant $\epsilon$.

\begin{theorem}{\bf (Lower Bound
    for \textsc{bounded threshold$-v$})}\label{pro:mightyLB}
    Consider the $\textsc{Threshold-v}$ problem in the noisy value model.
  Fix an arbitrary set of elements $X$ with values $V(X)=(v_1,v_2,\dots, v_n)$ such that the number of distinct elements $\ell$ is bounded by
  $\ell=O(n^{1-\varepsilon})$, where
  $\varepsilon >0$ is a constant.
  Let $k_v$ be the number of elements with value at least $v$.
  Any instance-optimal algorithm, even with prior knowledge of
   $V(X)$, and even with no constraints on the number of rounds,
  has an  query complexity 
  \[ q= \Omega\left( n+n \log (\min \{k_v,n-k_v\} ) \right). \]
  \end{theorem}
  \fnote{Is is weird to first quantify all instances and then algorithms? Is it okay not to talk about probabilities since it's in the definition of instance optimal algorithm?}

\subsubsection{Upper Bound (comparison model) - Proof of \autoref{pro:mightyUB}}

We start by providing an Algorithm  $\Threshold(X,v)$, which,  finds with probability at least $2/3$ all elements with value at least $v$. The query  complexity is $O(n\log k_v)$ and the round complexity is $O(\log^*n \cdot \log\log k_v)$, where we recall that $k_v$ is the number of elements with value larger than $v$.
For the ease of notation we will write $k$ instead of $k_v$.
Remember that the algorithm does not know $k$ in advance (otherwise the algorithm has asymptotically optimal  round complexity of $\log_k^*n$). To circumvent this, the algorithm guesses $k$ by starting with $2$ and from there on it quadratically increases its current estimate resulting in $\log \log k + O(1)$ rounds.

The algorithm is recursive and its idea is as follows.
At every recursive call with  parameter $k'$, which is the current estimate of $k$, the algorithm simply divides the input into $k'^2$ parts and finds the two maxima in each part and then verifies that these maxima are indeed larger than the threshold $v$. 
There are three possible outcomes for each part.
\begin{enumerate}[(i)]
\item The maxima were smaller than $v$. In this case no further queries in this part are made.
\item  Exactly one maxima is larger than $v$. This element will be part of the final output, and similarly as before, no further queries in this part are made.
\item The two maxima have values above $v$.
In this case there are potentially even more values above $v$ and thee algorithm calls itself recursively (in hope to find even more such values). At the same time the estimate of $k$ is increased from $k'$ to $k'^2$; 
the idea being that the deeper the level of the recursion, the small the parts become and larger estimates $k'$ of $k$ can be tolerated (without exceeding the query complexity).
\end{enumerate}
As we will see in the analysis, the depth of all leaves in the recursive tree is likely the same---up to an additive constant.

In order to achieve bounds of \autoref{pro:mightyUB} we need to handle the case that the number of elements with value larger than $v$ exceeds $n/2$ separately: We use a meta-algorithm that simply runs a symmetric algorithm in parallel; the symmetric version stops whenever it is sure with probability at least $5/6$ that there is a set $\bar X$ of  $n-k$ elements strictly smaller than $v$. 
   Should this algorithm stop before  $\Threshold(v, X )$, then we simply output $X\setminus \bar X$ otherwise we output whatever  $\Threshold(v, X )$  computes.
   Additionally, to ensure that the meta-algorithm is a Monte Carlo algorithm, it simply terminates whenever the desired query or round complexity of \autoref{pro:mightyUB} is exceeded.

\begin{algorithm}[H]
\caption*{$\operatorname{\bf Algorithm}$\ $\Threshold(X,v, (optional)~ k' )$     \  \ \ \ \ (see \autoref{pro:mightyUB})
\\
{\bf input:} set $X$, threshold $v$, a positive integer $k'$ (estimate of $k_v$, \ie the size of the output)\\
{\bf output:} all elements in $X$ with value larger than $v$\\
{\bf error probability:} $1/6$
}
\begin{algorithmic}[1]
\label{alg1}

\STATE $k' \gets \max\{ k', 2\}$
	\STATE  partition $X$ into $k'^2$ randomly chosen sets  $Y_1, Y_2, \dots, Y_{k'^2}$  of equal size

        \FORALL{$i \in [k'^2]$}
	\STATE $y^1_i \gets \Max(Y_i,O(\log_{2}^*(n)), \frac{1}{64k'^4})$
	\STATE $y^2_i \gets \Max(Y_i \setminus \{ y^1_i \} , O(\log_{2}^*(n)), \frac{1}{64k'^4})$
        \ENDFOR
\STATE $\tau \gets 100\log(64k'^{4})$
        \FORALL{$i \in [k'^2]$}
	\STATE test $\tau$ times whether $y^1_i$ and $y^2_i$ are at greater or equal than $v$
        \STATE $p^1_i,p^2_i \gets$ number of positive answers to the tests
        \ENDFOR

        \FORALL{$i\in [k'^2]$  }
	\STATE $X'_i =  \begin{cases}
		\emptyset & p^1_i < \tau/2 \text{ (no large element) }\\
		\{ y^1_i \} & p^1_i \geq \tau/2 \text{ and } p^2_i < \tau/2 \text{ (one large element)}\\
			\Threshold(v, X,  k'^2 )& p^1_i \geq \tau/2 \text{ and } p^2_i \geq \tau/2 \text{ (at least two large elements) }\\
	\end{cases}$       
        \ENDFOR
	\STATE Return $\bigcup^{k'^2}_{i=1}X'_i$ 
\end{algorithmic}
\end{algorithm}

The following lemma  proves   \autoref{pro:mightyUB}  for $k\leq n/2$. 
As noted above, by running a symmetric version of the algorithm in parallel (that finds elements smaller than $v$), one can  obtain matching bounds for $k>n/2$ yielding  \autoref{pro:mightyUB}. 

\begin{lemma}\label{lem:frank}
%
%
Consider the the noisy-value-model and the Threshold-$v$ problem.
Consider an arbitrary instance $(X,v)$.
Let $k$ be the number of elements above with value greater than $v$.
Algorithm $\Threshold$, without prior knowledge of the multiset $X$
nor $k$ finds the correct output w.p. at least $5/6$ and makes
$O\left( n \log k\right)$ value queries and uses at most $O(\log_2^*n \cdot \log \log k)$ rounds. 

\end{lemma}
\begin{proof}
 \noindent 
{\bf Correctness:} 
Consider the recursion tree, where the root (at depth $0$) is the initial call to $\Threshold(X,v)$.
Consider the recursions at a node $u$ of the recursion tree with parameters $v, X, k'$.
Let  $Y_1, Y_2, \dots, Y_{k'^2}$ be the partition of $X$ at node $u$.

We claim that for  every child of $u$ (in the recursion tree), corresponding to one of the part $Y_i$, all of the following holds  w.p. at least $1-1/(16k'^4)$. 
\begin{enumerate}
	\item If $Y_i$ contains no element larger than $v$, then $p_i^1<\tau/2$ (the found maximum was in at most $\tau/2$ tests larger than $v$).
	\item If $Y_i$ contains one element larger than $v$, then the value of $y^1_i$ is at least $v$,  $p_i^1 \geq \tau/2$ and $p_i^2 < \tau/2$.
	\item If $Y_i$ contains two elements larger than $v$, then 
	 the values of $y^1_i$ and $y^2_i$ are at least $v$, 
	$p_i^1 \geq \tau/2$ and $p_i^2 \geq \tau/2$.
\end{enumerate}

Assuming the above claim,
we can consider the error over all layers (depth of the recursive tree).
Note that at level $i$ of the tree, $k'$ is of size $2^{2^i}$ and the error probability equals
$1/(16k'^4)= \frac{1}{16\cdot \left(2^{2^i}\right)^2} $.
 Thus, by Union bound over all parts, the  error is at most
	  	   \[ \sum_{i=0}^{\text{number of layers}} \frac{2^{2^i}}{16\cdot \left(2^{2^i}\right)^2} 
		   = \sum_{i=0}^{\text{number of layers}} \frac{1}{16\cdot 2^{2^i}} \leq \frac{1}{16} \sum_{i=0}^{\infty} \frac{1}{2^{2^i}} \leq \frac{1}{8} .  \]
Note that the above claim yields the correctness since the claim ensures  that all elements above the threshold are returned. 
	  
We now prove the above claim.	  
	  Consider an arbitrary part $Y_i$.
	  We distinguish between three cases:
	  
	  \begin{itemize}
	  	  	\item $Y_i$ contains at least two elements larger than $v$:
	  W.p. at least $1-1/(64k'^4)$ the element $y^1_i$ returned by $\Max$ will exceed $v$ and w.p. at least $1-1/(64k'^4)$ half of the queries to it will be larger than $v$ (using Chernoff bounds \autoref{lemma:Chernoff}). 
	  Similarly, $y^2_i$ will be computed  correctly and at least half of the queries will be larger than $v$.
	  By Union bound w.p. $1-1/(16k'^4)$ node $u$ will correctly launch a recursive call to $Y_i$. 
	  	
	  	\item $Y_i$ contains exactly one element larger than $v$:
	  W.p. at least $1-1/(64k'^4)$ the element $y^1_i$ returned by $\Max$ will exceed $v$ and w.p. at least $1-1/(64k'^4)$ half of the queries to it will be larger than $v$. 
	  Furthermore, $p^2_i < \tau/2$ w.p. at least   $1-1/(64k'^4)$ since $y^2_i$ is smaller than $v$.  
	  By Union bound w.p. $1-1/(16k'^4)$ node $u$ will correctly return $y^1_i$ and not launch a recursive call to $Y_i$. 
	  \item  $Y_i$  does not  contain any element larger than $v$:
	  Similarly as before, by Union bound w.p. $1-1/(16k'^4)$ node $u$ will not launch a recursive call to $Y_i$.
	  \end{itemize}
This proves the claim and yields the correctness.	
  
	  \medskip
	   \noindent 
	  {\bf Round complexity:} 
Consider the recursion tree and assume that the computation at every node is correct.
At each level of the recursion tree $O(\log^*_{2} n)$ rounds are necessary (due to $\Max$).
Suppose the depth of the recursion tree was $\log \log k+2$. Then
The total number of rounds is at most $O(\sum_{i=0}^{\log \log k+2}\log^*_{2} n)$
 which is bounded by $O((\log \log k) \cdot \log^*_2 n)$.
It remains to show that depth of the recursions tree is $\log \log k+2$

 
Consider the level $\ell=\log \log k +2$. There are $2^{2^{\ell}}= 2^{4 \cdot 2^{\log \log k}}=k^4$ elements at that level.

 Instead of considering the process that divides elements into equal-sized parts, we consider the equivalent version that works as follows.
 The elements are assigned one after the other to the buckets:
  initially, each bucket has $n/k^4$ empty slots. After $t$ elements were assigned, there are $n-t$ slots left, each equiprobable. Thus, the probability for an element to be assigned to a bucket $j$ is given by $(n/k^4 - z_j)/(n-t)$, where $z_j$ denote the number of elements among the first $t$ rounds that were assigned to bucket $j$.

  Consider the set $V$ of elements with value at least $v$.
In order for the algorithm to initiate a recursive call at a node $u$ (of the tree), at least two elements of $V$ at the same node must be assigned to the same part (child of $u$).

Consider all parts $Y_1, Y_2, \dots, Y_{k^4}$  on level $\ell$ each containing $n/k^4$ elements.
 Note that every partition of the $n$ elements to the $k^4$ buckets is equally likely.
  At node $u$ there are at most $k$ elements of $V$. W.l.o.g. we assume there are exactly $k$.
  Let $Z_1, Z_2, \dots, Z_k \in [k^4]$ denote the buckets to which the $k$ elements are assigned to. 
 It is worth mentioning that the $Z_i$ are correlated. Despite this correlation, we can combinatorially bound the probability for two or more elements end up at the same node.
 Let $S$ be any subset of $[k]$. In the following we will condition  on
 an arbitrary assignment (event) $\bigcap_{j\in S} \{Z_j=z_j\}$ (possibly the worst-case assignment), meaning that some of the elements above the threshold are already assigned to buckets.
 In the worst-case all of them are in different buckets, decreasing the probability for the next element to be placed in a bucket without any elements.
 
 We have for all $i\in [k]$ with $i\not \in S$, we have that the probability of a `collision' is bounded by
 \begin{align*} \Pr{ Z_i \in \bigcup_{j \in S} \{  z_j  \} ~\middle|~ \bigcap_{j\in S} \{Z_j=z_j\} } \leq \frac{k\cdot (n/k^4-1) }{n-|S| }\leq\frac{n/k^3-k }{n-k }\leq \frac{n/k^3 }{n } = \frac1{k^3}.  
 \end{align*}
 
 Thus, by Union bound, w.p. at least $1-1/k^2$ all $k$ elements are in different buckets.
By Union bound over all errors, we get that w.p. at least $5/6$ the round complexity is bounded as desired.
  
	  \medskip
 \noindent 
  {\bf Query complexity:} 
  Note that the query complexity is only a function of the depth of the tree since the query complexity at given level of the recursive tree is fixed.
  Assuming that the computation at every node was correct, and that the depth of the tree is at most $\log\log k +2$ (which happens w.p. at least $5/6$, we can calculate the bound on the query complexity.
For some large enough constant $C$,
  we get that the number of queries is bounded by
  
  \begin{align*} \sum_{i=0}^{\log \log  k +2 } C \cdot n \log  \left( 2^{4\cdot 2^{i}} \right)  
   =\sum_{i=0}^{\log \log k +2 }4 C \cdot n  \cdot 2^{i} 
  \leq  8C \cdot n  \cdot 2^{\log\log k+2} 
  =    8C \cdot n  \cdot 2^{\log(4\log k)} 
= O(n \log k).
 \end{align*}
  \medskip
 \noindent 
  
%
%
%
Note that the error probability at level $i$ of the tree 
is at most $\frac{2^{2^i}}{16\left(2^{2^i}\right)^2}$. Hence the probability to continue to the level of the tree (after level $\log \log k +2$) is bounded by the same probability.
Thus 
\[ \E{q} \leq O(n \log k) + \sum_{i=\log\log k +2} C \cdot n \log  \left( 2^{4\cdot 2^{i+1}} \right)  \frac{2^{2^i}}{16\left(2^{2^i}\right)^2} = O(n \log k).
 \]

\end{proof}

\subsubsection{Lower Bound (value model) - Proof of \autoref{pro:mightyLB}}

Recall that $k=k_v$ denotes the number of elements whose value is greater than or equal to the threshold-$v$. 
 Without loss of generality we assume $k \leq n/2$. 
Assume, for a contradiction, that there exists an algorithm $A$ with success probability at least $2/3$ and worst case number of value queries\footnote{no effort was made to derive tight constants} at most $T= (n\log_3 k )/10^6$. For simplicity, we will assume that 
the number $\ell$ of distinct values  is at most $\sqrt{n}$. This is not crucial as the proof extends to any $ \ell \in [2, O(n^{1-\varepsilon})]$ for constant $\eps$---in which case $T$ becomes a function of $\varepsilon$.
Let $q' = (\log_3 k)/1000$. Our proof takes a detour through a different computation model, which we call the 3-phase model.

\paragraph{The 3-Phase Model($q'$).}

During the first phase, the adversary provides
the outcome of $q'$ noisy queries (correct with probability at least 2/3) for all the elements.

During the second phase, the adversary can provide for free the correct value of some of the elements, depending on 
the outcome of the first part.

During the third phase, the algorithm strategically and adaptively 
chooses $n/1000$ elements to query, and in response the adversary reveals their true value.  

Finally the algorithm chooses the output.

\begin{lemma}
  \label{claim:reductionthreshold}  \label{lem:reductopk}
  If there exists an algorithm in the noisy-value model 
  with query complexity $T=n q'$ and 
  with success probability at least 2/3, then there exists an algorithm 
  in the 3-Phase model($q'$) with success probability at least 2/3.
\end{lemma}

\begin{proof}[Proof of \autoref{claim:reductionthreshold}]
Assume that there exists an algorithm $A^*$ in the noisy-value model. 
  Our algorithm $A$ starts querying the adversary for each element $q'$ times in Phase 1, and stores  the outcome of those queries.
  Algorithm $A$ then discards any information provided by the adversary during Phase 2. 
  In Phase 3, at each step algorithm $A$ must decide which element to query next. To that purpose, it starts executing $A^*$ step by step with the following strategy to simulate the noisy-value model to provide answers to the queries of $A^*$. If the next element queried by $A^*$ has had so far
\begin{itemize}
\item   strictly less than $q'$ queries: then algorithm $A$ uses the answers stored from Phase 1. 
\item   exactly $q'$ queries: then algorithm $A$ makes a Phase 3 query. 
\item  strictly more than  $q'$ queries:  then algorithm $A$ reuses the previous answer for that element
  \end{itemize}
  Finally, it uses the result of algorithm $A^*$ for its output.
  
 In the third case, since this is the correct answer,  this respects the error guarantee of the noisy-value model  
  (note that in our noisy-query model, the adversary may choose to answer with the correct 
  answer).    
  
  Since $A^*$ makes
  at most $n q' /1000$ queries,  at most $n/1000$ elements  in the execution of $A^*$ are queried $q'$ times or more, so algorithm $A$ makes at most $n/1000$ Phase 3 queries during the simulation. 
 \end{proof}

From \autoref{claim:reductionthreshold} we now focus on proving the lower bound in the 3-Phase model.  
Assume, for a contradiction, that there is an algorithm $A$ that succeeds with probability at least 2/3 in
   the 3-Phase model($q'$).
We define the adversary so that, among the $k$ elements which form the correct output, there is one that hides almost uniformly among a constant fraction of elements.

\begin{lemma}
There exists a partition of the $n$ elements of the following form:
\begin{itemize}
\item
The $k$ elements with value greater than or equal to $v$, forming a set $S_0$ that is the only correct output.
\item
$r=\Theta(\sqrt{n})$ \emph{buckets} $S_1,S_2,\ldots,S_r$ such that in each bucket $S_j$, all elements have the same value $z_j$; and each bucket has cardinality in $[\sqrt{n}/10, 2\sqrt{n}/10]$
\item
$\Theta(n) $ other elements, forming a residual set $R$.
\end{itemize}
\end{lemma}
\begin{proof}
To create that partition, first set aside the elements greater than or equal to $v$ to define $S_0$; then, group the remaining elements by value, forming at most $\ell$ groups; sort those groups by non-increasing cardinality, and take the groups one by one in that order, repeatedly splitting any group of cardinality greater than $2\sqrt{n}/10$ as needed to create buckets $S_1,S_2,\ldots$, each of cardinality in $[\sqrt{n}/10, 2\sqrt{n}/10]$, until the total size of the buckets is between $4n/10-2\sqrt{n}/10$ and $4n/10$. Set $R$ consists of the rest of the elements.

Since the groups of cardinality less than $\sqrt{n}/10$ have total size at most $\ell\sqrt{n}/10=n/10$, the other groups have total size at least $n-k-n/10\geq 4n/10$, so this construction succeeds. The number $r$ of buckets is at least 
$(4n/10-2\sqrt{n}/10) / (2\sqrt{n}/10) \geq 2\sqrt{n}-1$ 
and at most 
$4n/10/(\sqrt{n}/10)=4\sqrt{n}$, so 
$r=\Theta(\sqrt{n})$, as claimed. The number of elements remaining in $R$ in the end is at least $n-k-4n/10\geq n/10=\Theta(n)$, as claimed.
\end{proof}

%
 %
\paragraph{Phase 1.} In Phase 1, the adversary uses the following strategy for its `lies'. 
For each element $i$ of $\cup_{j=1}^{r} S_j$, the adversary always answers the true value $v_i$.
For each element $i$ in $S_0\cup R$, pick uniformly at random a
  set $S_j$, $1\leq j\leq r$; when $i$ is queried, the adversary answers the correct value $v_i$ of $i$ 
  with probability $2/3$, and value $z_j$ with probability $1/3$.
This defines the adversary strategy in the first phase.

We say that a bucket $S_j$, $1\leq j\leq r$, \emph{received} an element  $i\in S_0\cup R$ if  $S_j$ was picked for $i$ and if all queries about $i$ in the first phase were answered with value $z_j$; then $i$ is called a \emph{liar}.
 The first phase is called {\it successful} (from the adversary's perspective) if at least one element of $S_0$ is a liar, and if every bucket $S_j$, 
$1\leq j\leq r$, received at least one element from $R$. 

\begin{lemma}\label{lemma:phase1} 
$\Pr{\hbox{Phase 1 is not successful}}\leq 2/100.$
\end{lemma}
\begin{proof}
To analyze Phase 1, observe that the probability that at least one element of $S_0$ is a liar equals $1-(1- 3^{-(\log_3 k)/1000})^k\geq 99/100$.  In addition we need every bucket receives at least one element. Fix a bucket $j$ (among the $r=\Theta(\sqrt{n}$ buckets) and an element $i$ among  the $\Theta(n)$ elements of $R$. Let $A_{ij}$ denote the event that $i$ is received by bucket $j$. The probability of $A_{ij}$ is at least 
$$\Pr{A_{ij}})\geq   \frac{1}{3^{    \log_3 (k) /1000   }}  \cdot \frac{1}{r} =\frac{1}{k^{1/1000}r}= \omega( \log n / n) .$$
By the Union bound and independence of $(A_{ij})$ for a given $j$, 
$$\Pr{ \forall j \exists i A_{ij}}=
1-\Pr{\exists j\forall i \overline{A_{ij}}}
\geq 1-\sum_j\Pr{\forall i \overline{A_{ij}}}=
1-\sum_j\prod_i  \Pr{ \overline{A_{ij}}}=
1-\sum_j\prod_i  (1-\Pr{ A_{ij}} ).$$
Thus 
$$\Pr{\hbox{Phase 1 is successful}}\geq \frac{99}{100}-r  (1-  \frac{ \omega(\log n )}{n}    )^{\Omega(n)}= \frac{99}{100}-O(\sqrt{n}) n^{-\omega(1)}\geq 98/100.$$
\end{proof}

\paragraph{Phase 2.} Assume Phase 1 is successful. In Phase 2, the adversary reveals the true values of all received elements of $S_0\cup R$ except for one per bucket, chosen so that exactly one of those unrevealed elements belongs to $S_0$. This defines the  second phase. We say that an \emph{apparent bucket} $S'_j$ consists of a bucket $S_j$ plus the hidden element it received.

\paragraph{Phase 3.}  To analyze Phase 3, let $i^*$ be the hidden element of $S_0$. 
\begin{lemma}\label{lemma:phase3}
$\Pr{i^*\hbox{ is found}} < 0.26.$
\end{lemma}
\begin{proof}
The total number of hidden elements found during Phase 3  is bounded by the sum $R_1+R_2$ of the following two terms. 

 $R_2$ is the number of hidden elements found in buckets in which at least $10\%$ of the elements have been queried in Phase 3: since the bucket size is at least $\sqrt{n}/10$, each such bucket must be queried at least $\sqrt{n}/100$ times. Since the total number of queries in Phase 3 is bounded by $n/1000$, we have $R_2\leq \sqrt{n}\cdot 1/10$.

$R_1$ is the number of hidden elements found in other buckets. Since the number of unrevealed elements in those buckets is at least $(9/10)\sqrt{n}/10$ at any time during Phase 3, and the hidden element is distributed uniformly among those, each query succeeds in finding the hidden element with probability at most $12/\sqrt{n}$, so $\E{R_1}\leq (12/\sqrt{n})\cdot (n/1000)$. By Markov's inequality, with probability at least $8/10$ we have $R_1\leq \sqrt{n}\cdot 6/100$. 

 The liar $i^*$ is distributed uniformly among the $r$ hidden elements, so the probability that $i^*$ is found given $R_1+R_2$ equals $(R_1+R_2)/r$. Thus
 $$\Pr{i^*\hbox{ is found}}\leq \frac{2}{10}+\frac{8}{10}\cdot  \frac{\sqrt{n}(1/10+6/100)}{2\sqrt{n}-1}  < 0.26.$$
 \end{proof}
 
 \paragraph{Wrapping-up.}
Taking into account  \autoref{lemma:phase1} and \autoref{lemma:phase3}, and the fact that if $i^*$ is not found then the correctness probability cannot exceed $O(1/n)$,  the overall success probability of the algorithm is at most $0.02+0.26+O(1/n)<2/3$, a contradiction.

\subsection{The \textsc{Top}-$k$ Problem}\label{sec:topk}
%
%
In this section, we present worst-case and approximation results for \textsc{Top}-$k$. When several input elements share the same value, there may be several correct outputs for $\Topk$. 
Throughout this section we assume $k\leq n/2$; otherwise one can simply look for the lowest-$k$ elements.

We start by giving tight bounds for the worst-case setting in terms of query complexity as well as round complexity.  

\begin{theorem}\label{thm:topksidekicks}{\bf (Algorithm for \textsc{Top}-$k$)}
Consider the \textsc{Top}-$k$ problem in the noisy comparison model.
  Fix a set of elements $X$ and  $\delta \in [1/n^7,1]$. \vincent{make sure the
    interval for $\delta$ is correct, and the dependency in the query
    complexity as well.}
  Algorithm $\Topk(k,X)$ computes the \textsc{Top}-$k$ elements of $X$ with correctness
  probability at least $1-\delta$ and 
  has 
  \begin{enumerate}
  \item  round complexity $r$ and 
  \item query
  complexity $q=O(n\log (kb/\delta))$,
  where $b$ is defined by $r=\log_{b}^*(n) +4$. 
  \end{enumerate} 
\end{theorem}

\begin{theorem}\label{thm:kmaxlower} {\bf (Worst-case Lower Bound for
    $\Topk$)} 
  Consider the \textsc{Top}-$k$ problem in the noisy value model. Let $k\leq n/2$.
  Any algorithm $A$ that is correct w.p. at least $2/3$
  has expected query complexity at least $\Omega(n \log k)$.
   Furthermore, suppose $A$  returns the \textsc{Top}-$k$ elements of $X$ with correctness probability at least
  $2/3$ in $r$ rounds 
  using at most $q=n \log_3 (k)$ value queries, then it must be that $r\geq \log_k^*(n)-O(1)$.

\end{theorem}


We then move to approximation algorithms.
In many practical settings (i) the round and query complexity are particularly
important and (ii) it suffices to return elements that are among the top elements,
but are not necessarily the top elements. 
For such a setting we propose algorithm $\mathbf{Approx-}\Topk$ and prove the
following theorem.

\begin{theorem}\label{thm:approx}
Consider the \textsc{top}-$k$ problem in the noisy comparison model. 
Let $X$ be a set of elements, $k$ an integer, $\delta \in [1/n^{5};1]$, and $\gamma > 0$.
Algorithm $\mathbf{Approx-}\Topk(X, k, r,\gamma)$ outputs
a set of $k$ elements that
belong to the top-$(1+\gamma)k$ elements of $X$ with correctness probability at least $1-\delta$ and has
\begin{enumerate} 
\item round complexity $r$ and
\item query complexity  $q=O\left(n \log \left(b\frac{1+\gamma}{\delta \gamma}\right ) + \frac{k}{\gamma} \log(kb/\delta) \right)$,  
  where $b$ is defined by $r=\log_{b}^*(n) +4$. 
\end{enumerate}
\end{theorem}
For example, for $r=O(\log^* n)$ and $\gamma=\Theta(1)$ the algorithm has query complexity $O( n + k\log k)$ and success probability $2/3$.
The worst-case bounds described in the above theorem are tight; a slight
modification of our worst-case lower bound for the maximum, implies
a lower bound for $k=1$ and any $\gamma = O(1)$ that matches the above
bounds.

\subsubsection{Upper Bound (comparison model) - Proof of \autoref{thm:topksidekicks}}
Our algorithm (see $\Topk$) is surprisingly simple. We partition the input into $k^4$ sets of equal size.
The \textsc{Top}-$k$ elements will be in distinct sets with sufficiently large probability.
For each of these sets we employ $\Max$ to  find w.p. $1-1/k^7$ the largest element with query complexity $O(n/k^4\cdot k^4\log k)=O(n\log k)$ and round complexity $\log_k^*(n)$ (which can be decreased by increasing the query complexity).
 
\begin{algorithm*}\caption*{$\operatorname{\bf Algorithm}$\ $\Topk(X,k,r,\delta)$  \  \ \ \ \ (see \autoref{thm:topksidekicks})
\\
{\bf input:}  set $X$, partition parameter $k$, number of rounds $r$, $\delta$ error probability \\
{\bf output:} largest $k$ largest elements in $X$\\
{\bf error probability:}  $\delta$.
}
\begin{algorithmic}
\label{alg:Topk}
\IF{$k/\delta > |X|^{1/12} $}
\STATE \Output\ {\bf 4-Round-Algorithm($X,k$)}  (\cite[Corollary 1]{BMW16})
\ELSE
\STATE partition $X$  into $k^*= 4k^4/\delta$ randomly chosen sets $Y_1,Y_2,\dots, Y_{k^*}$ of equal size
\FOR {$i \in [1,k^*]$ in parallel}
\STATE  \inparallel, $y_i\gets \Max(Y_i,r, \delta^2/(4k^7))$
\ENDFOR
\STATE \inparallel, compare  $y_i$ to $y_j$, exactly $1000 \log(k/\delta)$ times for each $i,j\in [1,k^*]$ with $i\neq j$
\STATE \Output the $k$ largest elements among those elements
\ENDIF

\end{algorithmic}
\end{algorithm*}

\begin{proof}[Proof of \autoref{thm:topksidekicks}]
	Without lost of generality, we assume that $k/\delta \leq |X|^{1/12}$ since otherwise the correctness follows directly from
	\cite[Corollary 1]{BMW16}.
	
 We  define the desired output by breaking ties as follows: whenever two elements share the same value, the one with the larger id is assigned  the larger rank. 

The algorithm can be incorrect for three reasons: 
\begin{itemize}
\item
It may be that several of the $\Topk$ elements fall into the same $Y_i$. However, we have that w.p. at least $(1-k\delta/(4k^4))^k\geq  1-\delta/(4k^2))$ that the \textsc{Top}-$k$ elements are in distinct $Y_i$.
\item
It may be that for some $Y_i$,  $ \Max(Y_i,r, \delta/(4k^7))$ fails to find the maximum. Let $\mathcal{E}_i$ be the event that $\Topk$ returns an  element  of $Y_i$ with the largest value among $Y_i$.
 We have $\Pr{\mathcal{E}_i} \geq 1-\delta/(4k^7)$ for all $i\in[1,k^*]$, by correctness of $\Max$(\autoref{lem:topksidekicks}).
\item
It may be that for some pair $y_i,y_j$, the majority of the comparisons of the second-to-last line yield is incorrect. We observe that the sorting in the second-last line fails only w.p. at most $\delta^11/(4k^11)$, by Union bound over all ${k^4/\delta \choose 2}$  comparisons.
\end{itemize}
Thus the probability that the output is correct is, by Union bound, at least \[1-\delta/(4k^2)-k^*\cdot \delta/(4k^7)-  (k^*)^2\cdot \delta^{11}/(4k^{11})\geq 1-\delta.\]

The round complexity is dominated by the calls to $\Max$ and thus the round complexity follows from 
\autoref{lem:topksidekicks}.

\end{proof}


\subsubsection{Lower Bound (value model) - Proof of \autoref{thm:kmaxlower}}

\begin{proof}[Proof of \autoref{thm:kmaxlower}]
	The  first part of the claim is a consequence of \autoref{pro:mightyLB}. Indeed, consider an arbitrary instance (multiset $V$) with $\ell=O(n^{1-\epsilon})$ distinct values, where $\epsilon>0$ is a constant, and such that the number of elements with value at least $v$ is equal to $k$. If we give free additional knowledge to the algorithm for $\Topk$, namely, the value $v$, then the problem is equivalent to Threshold-$v$ with the knowledge of $k$, the number of elements with value at least $v$. Since the lower bound of \autoref{pro:mightyLB} applies even with knowledge of the multiset $V$, it applies to the case where we know $k$, hence the first part of the claim.
	
	In fact, we can create an arbitrary instance with at most $\ell=O(n^{1-\varepsilon}),\varepsilon >0$ distinct values. Note that we assume that $k\leq n/2$. 
	
	Consider  the second part of the claim.
		The proof is along the same lines as the proof of \autoref{lemma:maxLBrankmodel}. The adversary strategy is as follows: when queried for the value of any element $x$ which is part of the $\Topk$ it answers with the true value.
	For any element $x$ 	which is not in the $\Topk$, by assumption there are at least $n/2$ many, with probability $2/3$ it answers with the true value of $x$, and with the complementary probability $1/3$ it answers with the value of rank $1$.
	At the beginning of iteration $t$, 
	let $S_{t-1}$ denote the set of elements whose queries have always answered with rank $1$ and that have received exactly as many queries as $x_1$, the element of rank 1.
	
	For $k \leq n/2$, using the same techniques as in \autoref{lemma:maxLBrankmodel}, we can show that
	$|S_t| = \Omega(\log n)$ making it highly unlikely to find the $\Topk$ elements. with fewer than  $\log_k^*(n)-O(1)$  rounds.
	
\end{proof}

\subsubsection{Approximation Algorithm - Proof of \autoref{thm:approx}}
The idea of the algorithm is to split the input into parts of equal size and to compute the maximum in each part.
In order to guarantee a low query complexity, the algorithm cannot afford to calculate the correct maximum in each part.
However, the query complexity is large so that \emph{enough} parts calculate the correct maximum:
Assuming that the $k(1+\gamma)$ top elements are spread almost equally among the parts, with good probability, more than $k$ parts containing top elements return the correct maximum. In a second step the algorithm finds accurately the \textsc{Top}-$k$ values among all maxima (of which there are not too many).

\FloatBarrier

\begin{algorithm*}[ht!]\caption*{$\operatorname{\bf Algorithm}$\ $\approxtopk(X, k, r, \gamma,\delta)$  \  \ \ \ \ (see \autoref{thm:approx})
\\
{\bf input:} set $X$, integer $k$, a number of rounds $r$, error probability $\delta$,  rank approximation factor $\gamma$\\ 
{\bf output:} $k$ elements with rank in $\{1, 2, \dots, k(1+\gamma) \}$\\
{\bf error probability:} $\delta$.
}
\begin{algorithmic}
\STATE $\gamma \gets  \min\{ 1,\gamma \}$
  \STATE  partition $X$ into $k^* =\frac{20k(1+\gamma)^2\log(1/\delta)}{\gamma}  $ randomly chosen  parts   $Y_1,Y_2,\ldots,Y_{k^*}$
    \STATE \inparallel, find the maximum element $y_i$ of $Y_i$
  for each $i$, in $r/2$ rounds with error probability
  $\frac{\gamma}{80\delta(1+\gamma/2)}$, using $\Max$.
  \STATE Compute and output the \textsc{Top}-$k$ elements of the set
  $\{y_1,\ldots,y_{k^*}\}$ in at most
  $r/2$ rounds and with error probability at most $\delta/6$, using
  $\Topk$.
\end{algorithmic}
\end{algorithm*}

\FloatBarrier
 \begin{proof}[Proof of \autoref{thm:approx}]
   The bounds on the query complexity follows directly
   from \autoref{lem:topksidekicks} and~\autoref{thm:topksidekicks}.
   More precisely, \autoref{lem:topksidekicks} shows that computing the maximum
   of a set of size $n_0 = n / k^*$
   in $r/2$ rounds
   with error probability at least  $\gamma\delta/(80(1+\gamma/2))$
   can be
   done using at most $O(n_0 \log (b\frac{1+\gamma}{\gamma\delta}))$,
   where
   $b$ is defined as $r/2 = \log^*_b(n)$. Summing up over the
   $k^*$ parts, yields a query complexity
   of $O(n \log (b\frac{1+\gamma}{\delta\gamma}))$
   for the first part and a round complexity
   of $r/2$.
   \autoref{thm:topksidekicks} shows that computing the \textsc{Top}-$k$ elements
   of a set of size $k^*$ in
   $r/2$ rounds
   with success probability at least
   $\delta/6$ can be done using 
   $O(k^* \log (kb/\delta))$ queries,
   where again $r/2 = \log^*_b(n)$.
   The query and round complexities of the algorithm follow.
   
 
   We thus aim at showing that the elements output by the algorithm
   are part of the top-$(1+\gamma)k$ elements of the input
   set with probability at least $1-\delta$.
   We say that an element is \emph{desirable} if it is part of a
   correct output
   to the top-$(1+\gamma)k$ problem on the input.
   We denote by event $\calE$ the event that
   there are at least $k(1+\gamma/2)$ parts $Y_i$ in the partition
   $Y_1,\ldots,Y_{k^*}$ that are such that
   the maximum of $Y_i$ is desirable.
   \begin{claim}
     \label{claim:nbdesirable}
     Event $\calE$ happens with probability at least $1-\delta/3$.
   \end{claim}
   \begin{proof}
     This follows from some classic result on ball-into-bins.
     Each desirable element is seen as a ball and is placed in a
     random part of $Y_1,\ldots,Y_{k^*}$. Each
     of these parts is seen as a bin. We assume that elements
     are placed into bins uniformly, one at a time.
     
     We aim at bounding the number of times a desirable element is
     assigned to a bin that already contains desirable elements.
     At any time $t$, the probability for placing the $t$th
     element into a bin already containing a desirable
     elements is at most the total number of desirable elements divided
     by the total number of bin, namely
     \[p' = \frac{(1+\gamma)k}{k^*} =
     \frac{\gamma }{20 \log(1/\delta) (1+\gamma)}. \]

     It follows that the expected  number of desirable elements that,
     at the end of the execution, are in a
     bin already containing a desirable element
     is at most $k \gamma/(10\log(1/\delta))$, since the algorithm ensures $\gamma \leq 1$. Note that we assume $\gamma \geq 1/k$ otherwise the problem reduces to finding the \textsc{Top}-$k$ elements.
     By Chernoff inequality (upper bounding each probability with $p'$), we have that event $\calE$ happens
     with probability at least  $1-\delta/3$.

     
   \end{proof}

   \begin{claim}
     \label{claim:approxcorrect}
     Conditioned on event $\calE$,
     the total number of desirable elements in the set (of maxima)
     $\{y_1,\ldots,y_{k^*}\}$ is at least $k$
     with probability at least $1-\delta/3$.
   \end{claim}
   \begin{proof}
     Since each execution of $\Max(S_i, r,
     \frac{\gamma\delta}{80(1+\gamma/2)})$
     on
     a set $S_i$ has probability at most
     $\frac{\gamma\delta}{80(1+\gamma/2)}$
     of not returning the maximum element of the set, the expected
     number of desirable elements that are not in the set 
     $\{y_1,\ldots,y_{k^*}\}$ is at most
     $\frac{\ell\cdot\gamma\delta}{80(1+\gamma/2)}$,
     where $\ell$ is the number of parts $Y_i$ that contain
     a desirable element. Since we conditioned on event
     $\calE$ happening and so $\ell \ge (1+\gamma/2)k$.

     Applying Markov inequality, 
     this number is at most $\frac{\ell\cdot\gamma}{2(1+\gamma/2)}$
     with probability at
     least $1-\delta/3$. It follows that with probability
     at least $1-\delta/3$, the total number of
     desirable elements in the set
     $\{y_1,\ldots,y_{k^*}\}$ is at least
     \[ \ell- \frac{\ell\cdot \gamma}{2(1+\gamma/2)} \ge  (1+\gamma/2)k - \frac{\gamma}{2}k =    k. \]
     
   \end{proof}

   We can now conclude the proof of the theorem. We condition
   on event $\calE$ happening. By \autoref{claim:nbdesirable},
   this happens with probability at least $1-\delta/3$. Thus, we can
   apply \autoref{claim:approxcorrect} and conclude that the
   set $s = \{y_1,\ldots,y_{k^*}\}$ contains
   at least $k$ desirable elements with probability at least
   $1-\delta/3$.
   Thus, conditioning on this last event,
   $\Topk$ outputs $k$
   elements among the top $(1+\gamma)k$ elements with success
   probability at least $1-\delta/3$.
   Therefore, taking a union bound over the probability that
   event $\calE$ does not happen, that the set $s$ does not
   contain $k$ desirable
   elements and the failure probability of $\Topk$,
   we have that $\approxtopk$
   is correct with probability at least
   $1-\delta$.

 \end{proof}


%

%

\section{On Instance Optimality (w.r.t. the Query Complexity)}
\label{sec:instanceopt}
\subsection{The Max Problem}
It's worth mentioning that \autoref{def:ionew} implies that the $\Max$ is instance-optimal since \autoref{lem:Omegan}
implies a lower bound of $\Omega(n)$ queries for any algorithm that is correct with probability at least 2/3.


\subsection{The Threshold-$v$ Problem}

From \autoref{pro:mightyUB} and \autoref{pro:mightyLB},
it follows that $\Threshold$ is instance-optimal (w.r.t.
the query complexity)
provided that the number of distinct values is polynomial
in $n$ (\textsc{bounded-value threshold$-v$}).
\begin{corollary}\label{cor:threshold-v:instanceopt}
  Algorithm $\Threshold$ is instance-optimal
  for the \textsc{bounded-value threshold$-v$} problem
  in both noisy-comparison and noisy-value models.
\end{corollary}

We complement this by the following result for the case of $\ell=\Omega(n)$ distinct values.

\begin{theorem}\label{thm:nofreelunch}
Consider the $\textsc{threshold-v}$ problem in the noisy value model. 
  There is no instance-optimal algorithm for threshold-$v$
  if the number of distinct values is $\Omega(n)$.	
\end{theorem}

Note that the important difference to  \autoref{thm:rankkupper} is that the algorithm does not know which values contain $2$ elements in the instance; again, if this was known than one could simply use a simple adaption of \autoref{thm:rankkupper}
to the given instance.

To proof constructs instances with $\Omega(n)$ distinct values, where a variant of the algorithm of \autoref{thm:rankkupper}  that knows the instance only requires $O(n + k\log n)$ queries; yet any algorithm not knowing the instances requires $\Omega(n\log k)$ queries.
\medskip
%
%

The above discussion implies that our results are not tight in the range where the number of distinct values is larger than any polynomial in $n$ (with exponent  strictly smaller than 1) and strictly sublinear in $n$; in symbols, $\ell \in [ \omega(n^{1-\eps }), o(n) ]$. We believe that the precise bound is inherently connected to the entropy.
\begin{conjecture}\label{conj:entropy}
Consider the $\textsc{threshold-v}$ problem in the noisy value model. 
Consider an arbitrary instance (multiset $V$).
Let $k$ be the index of the smallest element above the threshold (unknown to the algorithm).
There exists an instance-optimal algorithm, \emph{with} prior knowledge of the multiset $V$, that outputs the all elements above $v$ w.p. at least $2/3$ and uses  
\[O\left( \sum_{i} s_i  \log \left( \min\{s_i , k, n-k  \}\right)   \right)      \] queries, where $s_i$ is the number of elements with value $i$.
This becomes $O\left( \sum_{i} s_i  \log s_i  \right)    $  for $k\leq n/2$ being polynomial in $n$. Furthermore, requires an expected query complexity of 
\[\E{q}=\Omega\left( \sum_{i} s_i  \log \left( \min\{s_i , k, n-k  \}\right)   \right).      \]  \end{conjecture}

\subsubsection{Lower bound (value model) - Proof of \autoref{thm:nofreelunch}}
\begin{proof}[Proof sktech of \autoref{thm:nofreelunch}]
	Let $k=\sqrt{n}$. We define a family of inputs $\mathcal{I}$ with $7n/8$ distinct values with at most 2 elements for each value.  
	There exists an algorithm that, knowing the multiset of values, solves Threshold-$v$ on such inputs with query complexity $O(n+k\log n)$, which is linear here. The algorithm is an easy extension of the one for the case where all values are distinct  \autoref{thm:rankkupper}\footnote{The charging argument in the proof of \autoref{thm:rankkupper} changes slightly}.
We will show that, without the knowledge of the multiset of values, any algorithm requires at least $\Omega(n\log k)$ queries, which is $\Omega(n\log n)$ here, concluding the proof.  

Let $v_1>v_2>\cdots $ denote the $7n/8$ values. We randomly partition the bottom $n/4$ values into two sets of size $n/8$, $B$ and $C$. The instance has two elements for each value in $B$ and one for all other values. This defines the instance. \fnote{so we create an instance here, right?}

If there exists an algorithm with  query complexity
$n\log k/10^6$, then at most $n/1000$ elements are queried more than $\log k/1000$ times. Similarly to our previous lower bound proof, we consider the two-phase computation model where in the first phase every element is queried $\log k/1000$ times and in the second phase the algorithm adaptively chooses $n/1000$ elements and the adversary reveals their true value. We will prove a lower bound in that model, implying a lower bound in the original computation model. See \autoref{pro:mightyLB} for details. 

The  adversary strategy is as follows. Every element $u$ of the \textsc{Top}-$k$ elements chooses u.a.r. without replacement a value $X_u$ in $C$. \fnote{So there  elements in $C$ have unique values, right?}
Whenever queried, $u$ responds with $X_u$ w.p. $1/3$ and the truth otherwise.
The other elements always tell the truth. This defines the adversary strategy.

With probability at least $9/10$, there is at least one element $u$ of the \textsc{Top}-$k$ for which the adversary answered $X_u$ for all $\log k/1000$ queries. 
For the output to be correct, the algorithm needs to identify $u$ during the second phase.
At the beginning of phase 2, there is a set of $n/8+1$ values, namely, $B\cup \{ X_u\}$, each apparently shared by two elements forming a pair, with identical information about each, and the algorithm needs to find where $u$ is hiding. 
$X_u$ is uniform among the values of the set $B\cup C$. Even if at each query of the second phase the adversary reveals the true values of both elements in the pair, by Markov's inequality, the probability that the algorithm finds $u$ after $n/1000$ queries is less than $1/10$. 

\end{proof}

\subsection{The \textsc{Top}-$k$ Problem}
In this section we show for a large range of parameters---covering many practical settings---how knowledge of the instance can help.
Interestingly, in this range, the query complexity is completely governed by $n,k$, the number of values that are strictly larger than the $k$'th largest element, and the number of repetitions of the $k$'th largest element.

\begin{notation}
Let $V=\{ v_1,v_2,\dots, v_n\}$ be the multiset of values of the input in non-increasing order: $v_1\geq v_2\geq \cdots \geq v_n$.
Let $\ksg$ denote the number of elements with value strictly greater than $v_k$ and $\keq$ denote the number of elements with value equal to $v_k$.
Let $s=\ksg+\keq-k$ denote the \emph{slack}: among $\kappa$ only $k-\lambda$ have to be found and thus $s$ elements can be `ignored'.
\end{notation}


\begin{theorem}\label{theorem:LB}
  {\bf (Instance-optimal Lower Bound for \textsc{Top}-$k$)}
  Consider the\textsc{Top}-$k$ problem in the noisy value model. 
  Assume that $\keq \leq n^{\varepsilon/3}$ and that the number of distinct values in the input is $\ell=O(n^{1-\varepsilon})$, where $\varepsilon >0$ is a constant.
    Then the instance-optimal query complexity is 
$$\inf_{A\in\mathcal{A}}\sup_{I \hbox{ permutation of }V} \hbox{(expected query complexity of $A$ on $I$)} 
= 
  \Omega\left( n\log\left(
        \ksg +\frac{ \keq}{s+1} \right)  \right). $$ 
\end{theorem}

%


\begin{theorem}\label{theorem:obliviousalgorithm}
  {\bf
    (Instance-optimal Oblivious Algorithm for \textsc{Top}-$k$)}
Consider the\textsc{Top}-$k$ problem in the noisy comparison model.    
Fix a set of elements $X$ with values $V(X)=(v_1,v_2,\dots, v_n)$ and  an integer $k$.
   Algorithm $\mathbf{Oblivious-}\Topk(X,k)$, that does not have prior knowledge of $V(X)$, returns the \textsc{Top}-$k$ elements of $X$ with correctness 
  probability at least $2/3$ and has
  \begin{enumerate}
  \item 
   expected query  complexity  
  $ \E{q}=O\left( (n\log\left( \ksg +
        \frac{ \keq}{s+1} \right) + k^3 \right),$ and 
  \item expected round complexity  $ \E{ r}=O(\log\log( \ksg +   \frac{ \keq}{s+1} ) \log_{  \ksg+  \frac{ \keq}{s+1} }^*(n))$.
\end{enumerate}
\end{theorem}

In particular, for instances where there is a unique correct output,
i.e. $\ksg + \keq = k$, the \textsc{Top}-$k$ problem has query complexity at
least $\Omega(n\log k)$ and round complexity at least $\log^*_k(n)-O(1)$,
hence the algorithm from \autoref{theorem:obliviousalgorithm} is optimal.

\subsubsection{Lower Bound (value model) - Proof of \autoref{theorem:LB}}
We will prove a lower bound on the maximum number of queries rather than the expected number of queries. This is without loss of generality up to truncating the execution once the number of queries has exceeded 100 times the expectation, and subtracting $1/100$ to the correctness probability.

For simplicity, we will assume that $\varepsilon = 1/50$ (the proof extends to any constant $\eps$.)

First,  assuming $\ksg>0$ we prove the  $\Omega(n\log (\ksg))$ lower bound. We use the little birdie principle and assume that the adversary reveals for free all $\keq$ elements $\{ i: v_i=v_k\}$. In that setting the algorithm's task is equivalent to finding the $\ksg$ elements of value strictly greater than $v_k$, among $n-\keq\geq n/2$ since $\keq=O(n^{\epsilon/3})$. By \autoref{pro:mightyLB} (and \autoref{lemma:reduction}) that problem has complexity $\Omega(n\log\ksg)$, hence the proof of the first part.
Second,  if $\lambda =\Omega(\keq)$   then the Theorem follows from the $n\log(\lambda+1)$, so we assume that $\lambda=o(\keq)$.
If $s+1 > \keq / 1000$ then the Theorem follows from  the $\Omega(n)$ lower bound from \autoref{lem:Omegan}, so we assume that  
$ s+1  \leq {\keq}/{1000}.$ 

We will show the $\Omega\left(n \log \left( \frac{\keq}{s+1} \right) \right)$ lower bound by extending the proof of \autoref{pro:mightyLB}.
Assume, for a contradiction, that there exists an algorithm $A$ with success probability at least $2/3$
and worst case number of value queries\footnote{no effort was made to derive tight constants} at most $T= \frac{n 
  \log \left( \frac{\keq}{s+1} \right)  }{1000\cdot 101000}.$
Let $q'=101000 T/n= 
  \log \left( \frac{\keq}{s+1} \right) /1000$. Our proof takes a detour through a different computation model, which we call the  
\emph{3-phase} model. 
\medskip

\textbf{The 3-Phase Model($q'$) (for \textsc{\textsc{Top}-$k$}):}
During the first phase, the adversary provides
the outcome of $q'$ noisy queries (correct with probability at least 2/3) for all the elements.

During the second phase, the adversary can provide for free the correct value of some of the elements, depending on 
the outcome of the first part.

During the third phase, the algorithm strategically and adaptively 
chooses $n/101000$ elements to query, and in response the adversary reveals their true value.  

Finally the algorithm chooses the output.

\medskip

We note that the proof of \autoref{lem:reductopk}  similarly applies to the   \textsc{top$-k$} problem in the instance-optimal setting here, so we will prove our lower bound in the 3-phase model. 

\claire{The argument in the above sentence is awkward. Need to rewrite the Lemma so that it can be appealed to directly for both lower bounds.} 

We define the adversary so that, among the $k+s$ elements within which $k$ must be found to form a correct output, there are $101(s+1)$ that hide almost uniformly among a constant fraction of elements.

\begin{lemma}
There exists a partition of the $n$ elements of the following form.
\begin{itemize}
\item
A set $S_0$ containing the $\ksg+\keq$ elements of value greater than or equal to $v_k$.
\item
$\ell_3=\Theta(n^{1-\epsilon/2})$ buckets $S_1,S_2,\ldots ,S_{\ell_3}$ such that in each bucket $S_j$ all elements have the same value $z_j$; each bucket has size $s_j\in [n^{\epsilon/2}/10, 2n^{\epsilon/2}/10]$; and the total bucket size, $\sum_1^{\ell_3} s_j$, is in the range $[n/9,n/8]$.
\item
A residual set $R$ containing the $\Theta(n)$ other elements.
\end{itemize}
\end{lemma}
\begin{proof}
Up to artificially subdividing very large buckets, we may assume that $s_j\leq 2n^{\epsilon/2}/10$ for all $j\geq 1$; and we can then partition the elements with value strictly less than $v_k$ by taking groups of elements of equal value by order of non-increasing cardinality, stopping as soon as the total size exceeds $n/9$. To argue that all those buckets have size greater than or equal to $n^{\epsilon/2}/10$, we observe that the total size of the buckets with size smaller than or equal to $n^{\epsilon/2}/10$ is at most $\ell n^{\epsilon/2}/10=O(n^{1-\epsilon/2})$ since $\ell=O(n^{1-\epsilon})$. Moreover, $|S_0|=\ksg+\keq=O(n^{\epsilon/3})$ by assumption and since $\ksg=o(\keq)$. Hence the total size of buckets with size greater than $n^{\epsilon/2}/10$ is at least $n-O(n^{\epsilon/3})-O(n^{1-\epsilon/2})\sim n$. Finally, since we stop as soon as the total size of the buckets exceeds $n/9$, by definition of $\ell_3$, there at least $n-\keq-\ksg-n/8\geq n/4$ elements in $R$ since $\ksg\leq \keq = O(n^{\epsilon/3})$.
\end{proof}

We assume that the input permutation is chosen u.a.r. from all  permutations on $n$ elements. 


\paragraph{First phase.}
In Phase 1, the adversary uses the following strategy for its `lies'. 
For each element $i$ of $\cup_{j=1}^{\ell_3} S_j$, the adversary always answers the true value $v_i$.
For each element $i$ in $S_0\cup R$, pick uniformly at random a
  set $S_j$, $1\leq j\leq \ell_3$; when $i$ is queried, the adversary answers the correct value $v_i$ of $i$ 
  with probability $2/3$, and value $z_j$ with probability $1/3$.
This defines the adversary strategy in the first phase.

We say that a bucket $S_j$, $1\leq j\leq \ell_3$, \emph{received} an element  $i\in S_0\cup R$ if  $S_j$ was picked for $i$ and if all queries about $i$ in the first phase were answered with value $z_j$; then $i$ is called a \emph{liar}.
 The first phase is called {\it successful} (from the adversary's perspective) if at least $101(s+1)$ elements of $S_0$ are liars, if they are all received in different buckets, and if every bucket $S_j$, 
$1\leq j\leq r$, received at least at least one  element from $R$.

\begin{lemma}
  \label{lem:LB1phase}
 $ \Pr{\hbox{Phase 1 is not successful}} \leq  	1/10+o(1).$
\end{lemma}
\begin{proof}
The expected number of elements from $S_0$ that are received is 
\[ \keq3^{-q'}=(s+1) \frac{\keq}{s+1}  \cdot 3^{-(\log_3 (\keq/(s+1)))/1000} = (s+1) (\keq/(s+1))^{1-\frac{1}{1000}} \geq 900 (s+1) \]
using the definition of $q'$ and our assumption $ (s+1)  \leq{\keq}/{1000}$. By Chernoff bounds,  \claire{(Sketchy)} with probability at least $9/20$, the number of elements from $S_0$ that are received is at least  $101(s+1)$. In addition, since the number of buckets is $\ell_3=\Omega(n^{1-\epsilon/2})$ and $101(s+1)\leq \keq=O(n^{\epsilon/3})=o(\sqrt{\ell_3})$, with probability at least $9/20$ all of those $101(s+1)$ elements are received by different buckets.

Fix a bucket  $j$ (among the $\ell_3$ buckets) and an element $i\in R$. Let $A_{ij}$ denote the event that $i$ is received by bucket $j$.
The probability of $A_{ij}$ is at least 
\[ \frac{3^{   - \log_3(\keq/(s+1))/ 1000} }{\ell_3} \geq  
\frac{1}{\keq^{1/1000}\ell_3}= \frac{\Omega(1)}{n^{1-\epsilon/2+\epsilon/3000}} \]
by assumption on $\keq$ and our upper bound on $\ell_3$. 
As in the proof of \autoref{pro:mightyLB},
$$\Pr{ \forall j \exists i A_{ij}}\geq 1-\sum_j\prod_i  (1-\Pr{ A_{ij}} )
\geq 1-n (1-   \frac{\Omega(1)}{n^{1-\epsilon/2+\epsilon/3000}} )^{|R|}= 1-o(1) $$
since $|R|=\Theta(n)$. 
\end{proof}

\paragraph{Second phase.}
Assume Phase 1 is successful. In Phase 2, the adversary reveals the true values of all received elements of $S_0\cup R$ except for one per bucket, chosen so that exactly $101(s+1)$ of those unrevealed elements belong to $S_0$. This defines the  second phase. We say that an \emph{apparent bucket} $S'_j$ consists of a bucket $S_j$ plus the hidden element it received.

\paragraph{Third phase.}
We use the little-birdie principle to  simplify the third phase.
When the algorithm has queried more than a fraction 
  $1/505$ of the elements of an apparent bucket, the apparent bucket is \emph{revealed}: the adversary reveals for free 
   the correct value of \emph{all} the elements in the apparent bucket.

Let $Z$ denote the total number of hidden elements of $S_0$ found during Phase 3. 
\begin{lemma}
\label{lem:LB:eqr1r2}
$\Pr{Z\geq 100(s+1)}\leq \frac{1}{250}+\frac{1}{5}.$
\end{lemma}
\begin{proof}
$Z$ is bounded by the sum of two terms, $Z\leq R_1+R_2$, where: $R_1$ is the number of such elements found by querying individual elements, and $R_2$ is the number of such elements found because the adversary revealed the entire apparent bucket. If $Z\geq 100(s+1)$, then  $R_1\geq 50(s+1)$ or $R_2\geq 50(s+1)$.

 The $R_1$ elements revealed are a uniform random fraction of at most $1/505$ of the elements of each bucket. In expectation, the number of desirable elements thus revealed is  $\E{R_1} \leq   101(s+1) \cdot \frac{1}{505}=(s+1)/5$. By Markov's inequality, $\Pr{R_1>50(s+1)}\leq 1/250.$
  
%

  To bound $R_2$, we will now argue that the expected number of desirable elements found in this part is at most $\E{R_2}\leq 101(s+1) (9/100)$. 
  Since the total number of elements queried in Phase 3 is bounded by $n/101000$, the total size of the buckets revealed is at most $505n/101000=n/200$. 
Since every bucket has the same size to within a factor of 2, and the total size of the buckets is at least $n/9$, the number of apparent buckets revealed is at most 
$$2 \ell_3 \frac{n/200}{n/9}=\frac{9}{100} \ell_3 .$$ 
  By Markov's inequality, $ \Pr{ R_2 \geq 50(s+1)}\leq 1/5$. 
%
%
\end{proof}


\paragraph{Wrapping-up.}
We now bound the probability that  algorithm $A$ is correct.
Assume that Phase 1 is successful. 
After Phase 2, exactly $101(s+1)$ desirable elements remain to be revealed; but a correct output must identify all but $s$ of those. 
If during Phase 3 fewer than $100(s+1)$ desirable elements are revealed, then the algorithm needs to guess for at least one element for the output, which will only be correct w.p. at most $O(1/n)$. Revealing $100(s+1)$ desirable elements means $Z\geq 100(s+1)$. Thus
$$  \Pr{\text{Output correct}}\leq \Pr{\hbox{Phase 1 not successful}}+ \Pr{ Z \geq 100(s+1) } +O(1/n).$$
Using \autoref{lem:LB1phase} and \autoref{lem:LB:eqr1r2}, this is at most $1/10+1/250+1/5+o(1)<2/3$, a contradiction.  

		

\subsubsection{Upper bound (comparison model) - Proof of \autoref{theorem:obliviousalgorithm}}

We start by giving a non-oblivious algorithm: an algorithm that knows the parameter $\lambda$ and $\kappa$.

\begin{lemma}\label{theorem:nonobliviousalgorithm}
{\bf (Instance-optimal Non-oblivious Algorithm for \textsc{Top}-$k$)}
Consider the\textsc{Top}-$k$ problem in the noisy comparison model.    
Fix a set of elements $X$ with values $V(X)=(v_1,v_2,\dots, v_n)$ and  integers $k,\ksg,\keq,r,$.
There exists an Algorithm that takes $X,k,\ksg,\keq,r,\delta$ as input and solve the
\textsc{Top}-$k$ problem with success probability $1-\delta$ and
\begin{enumerate}
\item round complexity $r$ and
\item query complexity $q=O\left( n\log\left( b\left(\ksg +
      \frac{\keq}{s +1}\right)/\delta\right) +k^2 \log (kb/\delta)  \right)$,   where $b$ is defined by $r=\log_{b}^*(n) +4$.
\end{enumerate}  
\end{lemma}

We will use the following algorithm to prove \autoref{theorem:obliviousalgorithm} and \autoref{theorem:nonobliviousalgorithm}.
We call this algorithm, \emph{parameterized \textsc{Top}-$k$}. The parameters are $k$, $\lambda$, $\kappa$, and $\delta$.
Let  $\gamma = (\kappa-k+\lambda)/(k-\lambda)$. We assume $\gamma\geq 1/k $ since otherwise the algorithm can simply
execute $\Topk(X,k,r,\delta)$.
If $k\log k \geq n$, then simply call $\Topk(X,k,r,\delta)$. For the same reason we assume $s \geq 1$.
The algorithm is as follows.
\begin{enumerate}
\item Find the top $\lambda$ elements of $X$ using Algorithm $\Topk(X,\lambda,r/2,\delta/2)$.
\item Consider the remaining instance $X'$ (\ie $X$ minus the top $\lambda$ elements of $X$).  
  Apply the Approximation $\approxtopk(X, k-\lambda, r/2, \gamma, \delta/2)$ (\autoref{thm:approx}) to find a set of $k-\lambda$
  elements among the top-$(1+\gamma)(k-\lambda)$ elements of $X'$.
  
\item Output the union of the elements found at Steps 1 and 2.
\end{enumerate}

\begin{proof}
  We show that running parameterized \textsc{Top}-$k$ with parameters $\lambda$ and $\kappa$ yields
  \autoref{theorem:nonobliviousalgorithm}.
  \paragraph{Query and round complexity.}
  The first step takes $O(n \log (b\lambda/\delta))$ queries and $O(\log_b^* n)$ rounds according to \autoref{thm:topksidekicks}.
  The second step takes $O(n \log (b\frac{1+\gamma}{\delta \gamma}) + \frac{k}{\gamma} \log (kb/\delta))$ queries and $O(\log^*_{b} n)$ rounds
  according to \autoref{thm:approx}. Since $\gamma = (\kappa - k + \lambda)/(k-\lambda)$, we have that
  the query complexity of the second part is 
  \[O\left(n \log\left( \frac{k-\lambda}{\kappa - k + \lambda} \frac{b}{\delta}\right) + \frac{k}{\gamma} \log (kb/\delta)  \right)=O\left(n \log \left(\frac{\kappa}{s+1}  \frac{b}{\delta}\right)+ k^2 \log (kb/\delta) \right)),\] since $\gamma \geq 1/k$ and $\kappa \geq k - \gamma$. The round round complexity of the second round is $r/2$.
   Summing up leads to the claimed complexity.
  
  \paragraph{Correctness.}
  Observe that by \autoref{thm:topksidekicks}, the elements found at Step 1 are the correct ones with probability at least $1-\delta/2$.
  If this event happens, then to solve the \textsc{Top}-$k$ instance, one only needs to find $k-\lambda$ elements in the instance
  $X'$. Moreover, by definition of $\kappa$, any subset of size $k-\lambda$ among the top $\kappa$ elements of $X'$
  would be a correct output. By \autoref{thm:approx}, the call to Algorithm $\approxtopk$ at the second
  step indeed yields a set of size $k-\lambda$ containing elements among the top-$(1+\gamma)(k-\lambda)$ elements of
  $X'$ with probability at least $1-\delta/2$. By definition of $\gamma$, this set only contains elements from the top-$\kappa$
  elements of $X'$.
  Therefore, the output is a correct solution to the \textsc{Top}-$k$ problem with probability at least $1-\delta$.
\end{proof}

We now turn to the design of an oblivious algorithm that will use as a black box our
non-oblivious algorithm.
Recall that the non-oblivious algorithm, knowing $\lambda$ and $\kappa$ has a query complexity of $n \cdot \log \left(\max \left\{ \ksg,z \right\}\right)$, $z=  \frac{\keq}{s +1}$ up to constants.
The first term is due to finding those elements that are strictly larger than the $k$'th largest element and the second term, $z$, comes from finding $k-\ksg$ elements among the $\kappa$ possible ones.

The idea of the algorithm is to estimate the maximum budget $\widehat{b}$ that is larger than the max of both these quantities. 
It turns out, but calculating the max of the found values and the min of the remaining values in one can efficiently test if a solution is correct.
Thus the algorithm can efficiently test if $\widehat{b}$ is larger than the maximum. 
To ensure that the query complexity is not exceed, we increase the guess of $\widehat{b}$ doubly exponentially, allowing us to bound the query complexity by a geometric series and also guaranteeing that in the final guess of $\widehat{b}$ causes only a constant blow-up in the query complexity. 

Now, once we are in an iteration where $\widehat{b}$ exceeds the max of $\ksg$ and $z$, something interesting happens:
we can simply set our estimate $\widehat{\ksg}$ to $\widehat{b}$, \ie overestimating $\ksg$. As we will argue, this causes also at most a constant factor increase of the query complexity.
Equipped with an overestimation  $\widehat{\ksg}$, we only need to find $k-\widehat{\ksg}$ remaining elements.
We can simply choose our estimate of $\kappa$, namely $\widehat{\kappa}$, such that it is the smallest possible value that does not blow up our query complexity. In other words, we underestimate $\kappa$ without exceeding the query complexity.

\begin{algorithm*}\caption*{$\operatorname{\bf Algorithm}$\ $\textbf{Oblivious-}\Topk(X,k)$   \  \ \ \ \ (see \autoref{theorem:obliviousalgorithm})
\\
{\bf input:}  set $X$, parameter $k$\\ 
{\bf output:} largest $k$ elements in $X$   \\
{\bf error probability:} $1/3$
}
\begin{algorithmic}
\IF{$k \log k \geq n$}
\STATE \Output $\Topk(X,k)$ and HALT
\ENDIF
\STATE $C\gets $ the constant in the query complexity of $\Topk$.
\FOR{ $i \in [2,\log\log k +1] $}
\STATE $\widehat{b} \gets 2^{2^i}$ (budget estimate; budget for this iteration will not exceed $O(n \log(\widehat{b}))$)

\STATE $\widehat{\lambda} \gets \widehat{b} $
\STATE $\widehat{\kappa} \gets  \min \left\{  k' \colon  \frac{k'}{k'-k+\widehat{\lambda}+1} \leq \widehat{b} \text{ and }  k' + \widehat{\lambda}\geq k  \right\}$ (budget not exceeded and $k$ values will be found)
\STATE $S\gets $ result of Parameterized Algorithm for \textsc{Top}-$k$ (\autoref{theorem:nonobliviousalgorithm}), with  
parameters \\\hspace{0.9cm}$k ,\widehat{\lambda}$,  $\widehat{\kappa}$, $b=\widehat{b}$, $\delta = 1/(16\widehat{b})$
\IF{ $\Min(S, \log_2^*n, 1/(10\widehat{\ksg}) ) >\Max(X\setminus S,\log_2^*n, 1/(10\widehat{\ksg}) )$  }
\STATE \Output $S$ and HALT 
\ENDIF
\ENDFOR
\STATE \Output FAIL
\end{algorithmic}
\end{algorithm*}

\begin{proof}[Proof of \autoref{theorem:obliviousalgorithm}]
	Similar to the proof of \autoref{lem:frank}, by geometric series and due to the super-exponential growth,
	all calls of $\Topk$ will return the correct output w.p. at least $9/10$.
	We condition on this and on $\Max$ and $\Min$ always returning the correct value.
	Conditioning on $\Max$ and $\Min$ always returning the correct value, ensure that unless the output is FAIL, no
	invalid set is ever returned. 
		It remains to argue that the correct output is determined before the for loop is finished.
		
     Consider the iteration $i^*$ where $\widehat{b}  \geq  \max\{ \lambda , \frac{\keq}{\keq - k + \ksg +1} \} =  \max\{ \lambda , \frac{\keq}{\zeta} \}$. 
   Thus, it holds that in this iteration we overestimate lambda, \ie $\widehat{\lambda} = \widehat{b} \geq \lambda$. Note that overestimating $\lambda$ help in the sense that the additional elements found are (by conditioning on the correctness of $\Topk$) all part of the output. The only problem of overestimating $\lambda$ is a larger query complexity, but as we will, the over-estimation arriving under the conditioning will only increase the query complexity by a constant factor.
   
    Overestimating $\lambda$ will thus decrease the number of repetitions of the $k$-largest value we seek to $k-\widehat{\lambda}$. Therefore  $k-\widehat{\lambda}$ out of $\kappa$ elements have to be found, without knowing $\kappa$. The algorithm simply sets its estimate $\widehat{\kappa}$ to the smallest value possible ensuring that query complexity is not exceeded. 
 In particular, in the iteration $i^*$ (where $\widehat{b}  \geq  \max\{ \lambda , \frac{\keq}{\zeta} \}$), we have that $\kappa$ fulfills 
 the inequalities \[  \frac{k'}{k'-k+\widehat{\lambda}+1} \leq \widehat{b} \text{ and }  k' + \widehat{\lambda}\geq k ,  \]
 using that $\widehat{\lambda}\geq \lambda$. However, $\kappa$  might not be the minimum value and $\widehat{\kappa}$ is in iteration $i^*$ an underestimate of $\kappa$, \ie $\widehat{\kappa} \leq \kappa$. Choosing a smaller value can only increase the query complexity, but for similar reasons as above, the increase of the query complexity is only constant. 
   
     Therefore the call $\Topk(X,k,\widehat{\keq}, \widehat{\ksg} )$ will return the correct result. 
     
%
     The round complexity can be bounded as follows.
    Let $x= \max\{ \lambda, \frac{\keq}{s+1}  \}  <k.$
    Conditioning on all the events above we have
       \begin{align*} \sum_{i=0}^{\log \log  x +1 }  \left(  n \log  \left( 2^{ 2^{i}} \right) + k^2 \log ( k 2^{ 2^{i}}  )   \right)  &\leq 
       \sum_{i=0}^{\log \log  x +1 }  \left(  n  2^{i} + k^2 \log ( k  )  + k^2 2^{i}   \right)\\
      &\leq  2\cdot n  \cdot 2^{\log\log x + 1}  + k^2 \log ( k  )(\log\log x +1)  + 2 k^2 2^{\log\log x + 1} \\
&= O(n \log x + k^3).
 \end{align*}
      Note that the expected round complexity is up to additive constants of the same order:
      Similar as in \autoref{lem:frank}, the query complexity at iteration $i$ is $O\left( n\log  \left( 2^{ 2^{i}} \right)\right)$, whereas the error probability decreases super-exponentially   and hence the expected query complexity is  bounded by $O(n\log x)$.  
\end{proof}


\FloatBarrier

%
\section{The Value Model is Strictly 'Easier' than the Comparison Model}\label{sec:rank}
%
In this section, we show that for some problems the
value model is strictly easier than the comparison model:
For some problems, including rank-$k$, there exists an algorithm in the value model with a query complexity
that is much lower than the query complexity required to solve the problem in the noisy comparison model.

More concretely, we show that
\begin{enumerate}
\item
 the 
\textsc{rank-$k$} problem has query complexity is (i)
$O(n + k \log n)$ in the noisy value model and (ii)
$\Omega(n \log k)$ in the noisy comparison model.
Recall that the input for the rank problem is a set of distinct
elements and the goal is to find the \textsc{Top}-$k$ elements among them.

\item Any problem that can be solved in the the value model can also be solved in the comparison model (with only a constant blow-up in the query complexity)

\end{enumerate}

\begin{theorem}{\bf (Efficient Algorithm
    for Rank-$k$)}
  \label{thm:rankkupper}
  Consider the $\textsc{rank-k}$ problem in the noisy value model. 
  Fix a set of elements $X$ and an integer $k$.
  Algorithm $\textbf{Distinct-}\Topk(X,k)$ returns the top $k$ elements with
  success probability $2/3$ and has
  query complexity  $q=O(n + k \log n)$.
\end{theorem}

\begin{theorem}\label{thm:uniquelower}{\bf (Lower Bound for \textsc{Top}-$k$)}
Consider the\textsc{Top}-$k$ problem in the noisy comparison model.   Let $A$ be an algorithm for the rank-$k$ problem
  with success probability at least $2/3$. 
  Then, the query complexity of $A$ is at
  least $\Omega(n \log k)$.
  
\end{theorem}

%
%
%

We would like to point out that for several of our lower bounds, the
model is as follows: with probability $2/3$ the answer is correct, with
probability $1/3$ the adversary picks an arbitrary answer, possibly the
correct one if it makes things harder for the algorithm.
\autoref{thm:uniquelower} holds even if the adversary has less power
and is forced to provide an incorrect answer with probability 1/3.

\subsection{Rank-$k$ is Strictly Easier in the Value Model }

\subsubsection{Upper Bound (value model) -
  Proof of \autoref{thm:rankkupper}}

We now describe the algorithm. 
We will query each element in blocks of increasing length to which we refer to as super-queries.  
The subroutine \textsc{Super-Query($x,\mu$)}  works as follows. Query $x$ exactly $\mu$ times and return the most frequent answer (ties broken arbitrarily).

Let $n_\ell(x)$ denote the number of queries to element $x$ after $\ell$ super-queries. 
We will ensure that $n_\ell(x) = 12 \sum_{i=1}^{\ell-1}  2^i = 12\cdot (2^\ell - 1)$.
Let $v_\ell(x) $ be the most frequent response (value) after $\ell$ super-queries (ties broken arbitrarily). 
We say that there is a \emph{collision} between two elements $x,y$ if
$v_\ell(x) = v_\ell(y)$. In words,
the most frequent value return for $x$ and $y$ is the same.

\begin{algorithm}\caption*{$\operatorname{\bf Algorithm}$\ $\textbf{Distinct-}\Topk(X,k)$\  \ \ \ \ (see \autoref{thm:rankkupper})
\\
{\bf input:}  set $X$, integer $k$\\
{\bf output:} \textsc{Top}-$k$ elements\\
{\bf error probability:}  $1/3$
}
\begin{algorithmic}
 \FOR{element $x$}
  \STATE $n(x) \gets  12$
  \STATE $v(x)\gets$ \textsc{Super-Query($x,n(x)$)}
  \ENDFOR
  \REPEAT 
   \IF{there exists $x,y, x\neq y$ such that $v(x) = v(y)$}
  		\IF{ $n(x) \leq n(y) $ }
  		  		\STATE $n(x) \gets 2 n(x)$
  		\STATE $v(x)\gets$ \textsc{Super-Query($x,n(x)$)}
  		\ELSE
  		  		\STATE $n(y) \gets 2 n(y)$
  		\STATE $v(y)\gets$ \textsc{Super-Query($y,n(y)$)}
    	\ENDIF
  \ELSE
  		 \STATE take an element $x$ with $v(x) \in \{ v_1, v_2, \dots, v_k\} $ such that $n(x)< 20 \log n$  
  		\STATE $n(x) \gets 2 n(x)$
  		\STATE $v(x)\gets$ \textsc{Super-Query($x,n(x)$)}
  \ENDIF
  \IF {$\sum_{x} n(x) > 34000(n +k \log n)$}
  \STATE \Output FAIL
  \ENDIF
  \UNTIL {for every $i\in[k]$ there exists exactly one element $v(x)$ with $v(x)=v_i$ and $n(x)\geq  20 \log n$  }
  \STATE \Output all $\{ x \colon v(x) \in  \{ v_1, v_2, \dots, v_k\} \text{ and } n(x) \geq 20 \log n\}$
\end{algorithmic}
\end{algorithm}
\FloatBarrier


\begin{proof}[Proof of \autoref{thm:rankkupper}]
   We bound the query complexity of the super-queries by constructing a coupling to sequences of $0/1$ variables which correspond to  the lies and true responses to queries. This allows us to abstract away from the exact content of the lies (values). For the coupling to be well-defined, we allow the adversary to adaptively choose a value whenever there is a lie.
   
   For every element $u$ of the input consider an infinite sequence $S_u(1),S_u(2), S_u(2), \dots$, where each $S_u(i)$ is drawn i.i.d. as follows.
   \[S_u(i)= \begin{cases} 
	0 \text{ (truth) } & \text{ w.p. $2/3$} \\
	1 \text{ (lie) } & \text{ otherwise} \\
 \end{cases}
\] 
Observe that the probability that a sequence of length $12 \ell$ contains more ones (lies) than zeros is bounded by $3^{-\ell}$: let $X$ denote the number of lies among the $12\ell$ queries. We have $\E{X}=4 \ell$.
Thus, by Chernoff inequality, $\Pr{X\geq \frac{3}{2} \E{X} } \leq \exp \left( -\frac{3}{2} \frac{\E{X}}{3}   \right)\leq 3^{-\ell}$.

Consider a sequence $\Xi_u$ consisting of infinitely many sequences $\Xi_u(1), \Xi_u(2), \Xi_u(3), \dots$ (in this order),  where $\Xi_u(i)$ is of length $12\cdot  2^{i-1}$. 

Divide the sequence $S_u(1),S_u(2), S_u(3), \dots$ into  subsequences $\Xi_u(1), \Xi_i(2), \Xi_u(3), \dots$.
Let $X_u$ denote the first subsequence $\Xi_u(i)$ from which on all subsequences starting from $\Xi_u(i)$ to $\Xi_u(n^4)$ contain more zeroes (representing truthful responses) than ones; we set $X_u \infty$ if no such sequence exists.
Let $\mathcal{E}$ be the event that $X_u < n^4$ for  all $n$ elements. Note that $\Pr{\mathcal{E} } \geq 1-n^{-4}$.
In order for $X_u > \ell$ to hold, it must be that at least one of the sequences $\Xi_u(j),\in [\ell,n^4]$ contained more zeros than ones. Thus, by Union bound,
\[ \Pr{X_u > \ell}  = \sum_{i=\ell}^{n^4} \frac{1}{3^i} \leq \frac{2}{3^\ell} . \]

Let $Y_u$ denote the total length of all subsequences up to (including)  $X_u$. 
From the above we get, by law of total expectation,
\begin{align*}
 \E{Y_u~|~ \mathcal{E}} &= \sum_{\ell=0}^{n^4} \E{Y_u~|~X_u=\ell}\Pr{X_u=\ell}  \leq 
\sum_{\ell=0}^{n^4} \E{Y_u~|~X_u=\ell}\Pr{X_u > \ell - 1} \\
&\leq
12 + \sum_{\ell=1}^\infty 12 \cdot (2^{\ell}-1) \cdot \frac{2}{3^{\ell-1}}  \leq 200.
\end{align*}

Therefore, by linearity of expectation and Markov inequality, 

\begin{align}\label{eq:goodqueries} \Pr{ \sum_{i=1}^n {Y_i} \geq 10 \cdot 200 \cdot n ~\mid~\mathcal{E}} \leq \Pr{ \sum_{i=1}^n {Y_i} \geq 10 \E{ \sum_{i=1}^n {Y_i}~|~\mathcal{E}} ~\mid~\mathcal{E} } \leq \frac{1}{10}.	
\end{align}

 The connection of the sequence $\Xi_u$ to element $u$ is as follows.
 After the $X_u$'th super-query to $u$, all further iterations return their true value. 
Therefore, there exists a coupling between the responses to $u$'th queries (and super-queries) and the 
infinite sequence $\Xi_u$; in particular, we are only interested in a prefix of $\Xi_u$.

\medskip
 We call all super-queries to $u$ after the $X_u$'th super-query \emph{bad}.  
 Intuitively, they are bad because $u$ has already revealed its true value.
 We call all other super-queries before and including the $X_u$'th super-query \emph{good}.
 
 The total query complexity $T$ of the algorithm is the sum of the query complexity due to good queries $G$ and bad queries $B$.
 By \eqref{eq:goodqueries}, the query complexity due to good queries is bounded by $O(n)$ w.p. at least $1-9/10-1/n^4$, by Union bound.
 Note that bad queries can only happen for two reasons.
 \begin{enumerate}[(i)]
 	\item verifying the identity of the elements pretending to be part of the \textsc{Top}-$k$ values and
 	\item whenever there is a collision for some value $v$, \ie two or more elements pretend to have the same value $v$.
 \end{enumerate}
 We can bound the query complexity of the bad queries due to (i) by $80 k \log n$ -- as only the elements of \textsc{Top}-$k$ generate
bad queries that way.
 For (ii), we use the crucial property in the rank-$k$ problem  that there can only be one such element that truly  has 
 value $v$. In other words, when we query the true element $x$ with value $v$, then this is because an element $y$ that pretends to have value $v$ was queried and has $n(y)\geq n(x)$.
 Hence, whenever raise the counter of $x$ (the rightful element) to say $n(x)$,
 the counter of $y$ must have been at least $n(x)/2$. 
 A simple charging scheme  shows $B\leq 16 G + 80 k \log n$. 
	Thus, the total query complexity is bounded by $T=B+G \leq 17G +80 k \log n \leq 17(2000n) +80 k \log n 
	\leq   34000(n+k\log n) $ due to \eqref{eq:goodqueries} w.p. at least $9/10$.

	Furthermore, 
	each element that is part of the output is queried at least $20 \log n$ times. Thus,
	the probability that an element with value $v=v_i, i\in [k]$ is not part of the output is $1/n^2$, by Chernoff bounds (\eg \autoref{lemma:Chernoff}). 
	This shows that no incorrect element is part of the output; it remains to show that all \textsc{Top}-$k$ elements are part of the output.
	Observe that in the verification part, whenever there are fewer than $k$ elements that pretend to be  part of the \textsc{Top}-$k$, then,  by  pigeonhole argument, we must have a collision. This ensures termination as the algorithm never reaches a state where there is no collision and  fewer than $k$ elements that pretend to be  part of the \textsc{Top}-$k$.
	
	 Union bound over all elements  and taking all other sources of error into account 
	 (i.e, the error the algorithm terminates prematurely due to too many queries and the probability that $\mathcal{E}$ does not hold) yields that the total error probability is bounded by $1/10 + n/n^2 + \Pr{\bar{\mathcal{E}}} \leq 1/3$.
	
\end{proof}
We note that the our algorithm is not round-efficient. It remains an open problem to find an algorithm using $O(n + k\log n)$ queries having 
a good round complexity. 

\subsubsection{Lower Bound (comparison model) - Proof of \autoref{thm:uniquelower}}

Our lower bound (\autoref{thm:uniquelower}) is a generalization of the decision tree technique from
Braverman, Mao and Weinberg's paper \cite{BMW16} to $k>1$. 

Assuming that the number of queries is $o(n \log k)$, we can show that no algorithm
can distinguish between the ``true'' underlying 
permutation of the input
and some permutations whose set of $k$ largest elements differs by one element. Hence, we
can show that any algorithm is more likely to output
an incorrect partition.

The proof is by contradiction. 
Consider an input permutation $\pi$, where $\pi_i$ denotes the element in position $i$. 
The goal of $k$-max is to output $\{ \pi_1,\pi_2,\ldots ,\pi_k\}$.
Assume that each query has error probability $1/3$. 

Up to symmetry we may assume that $k\leq n/2$.
Let $A_0$ be an algorithm for $k$-max with expected number of queries at most $\ell_0$ and probability of being correct at least $(1-\delta_0)$ where 
\[\ell_0 = (c/\delta_0)\beta(1-2\delta_0) (n-k) \log_2\left(\frac{81k}{(1-2\delta_0)^2}\right)\]
 for $\beta:=1/3$ a fixed constant\footnote{For $k$-max, $\beta$ must be such that $1-2\beta$ is a positive constant. 
 } and $c$ some small enough constant, so that $\ell_0=\Theta(n\log (k/(1-2\delta_0)))$.

Let $\ell:= \delta_0 \ell_0$ and $\delta:=1-2\delta_0$.
By Markov's inequality (to prune long executions) and padding (to lengthen short executions), there exists an algorithm $A$ that has probability of being correct at least $\delta$ and uses exactly $\ell$ queries on every input.

Consider an input permutation $\pi$ and a root-to-leaf execution path $L$ in the decision tree of $A$. For any $(i,j)\in [k+1,n]\times[1,k]$, 
we define $G(i,j)$  (resp. $B(i,j)$)  to be the number of comparisons between element $\pi_i$ and elements in $\{ \pi_{j},\pi_{j+1},\ldots, \pi_{i-1} \}$  that have the correct (resp. incorrect) outcome. Let $W$ denote the event that there are at most $2\beta (n-k)k$  pairs $(i,j)\in [k+1,n]\times [1,k]$ such that $G(i,j)-B(i,j) > \gamma$, where  $\gamma := 2\frac{1}{(n-k)}\frac{1}{\beta}  \frac{4}{\delta}  \frac{\ell}{3}  =  \frac{8c\log_2(81k/\delta^2)}{3}=\Theta(\log(k/\delta))$.

\begin{lemma}\label{lemma:W}
For any permutation $\pi$,  $\Pr{W~|~\pi}\geq 1-\delta/2$.
\end{lemma}
\begin{proof}
We generalize the proof presented in \cite{BMW16}. Throughout this proof, permutation $\pi$ is fixed and all probabilities and expectations are implicitly conditioned on $\pi$.

How does $G(i,j)-B(i,j)$ change during the execution of $A$ on input $\pi$? Initially $G(i,j)=B(i,j)=0$. Whenever there is a comparison between $\pi_i$ and an element $\pi_{j'}$ with $j\leq j'\leq k$, then with probability $2/3$ the result is correct and $G(i,j)-B(i,j)$ increases by 1; with the complementary probability $1/3$ the result is incorrect and $G(i,j)-B(i,j)$ decreases by 1. 

Equivalently, the answer to a query is random unbiased (equally likely to be correct or incorrect) with probability $2/3$, and correct with the complementary probability $1/3$. Thus $G(i,j) - B(i,j)$ has the same distribution as $P(i,j) + Q(i,j)$, where 
we define variables
$P(i,j)$, $Q(i,j)$ as follows:
Initially, $P(i,j)=Q(i,j)=0$. Whenever there is a comparison between $\pi_i$ and an element $\pi_{j'}$ with $j\leq j'\leq k$, then with probability $1/3$, $Q(i,j)$ increases by 1, with probability $1/3$, $Q(i,j)$ decreases by 1, and with the remaining probability $1/3$, $P(i,j)$ increases by 1.

Each comparison query between $\pi_i$ and $\pi_{j'}$ for  $(i,j')\in [k+1,n]\times [1,k]$
affects $P(i,j)$ for at most $k$ values of $j$, namely, $j \in  [1,j']$.
Since $A$ uses $\ell$ queries, the variables $P(i,j)$, $(i,j)\in[k+1,n]\times [1,k]$ are considered at most $\ell k$ times in total, and so
$
\E{ \sum_{i=k+1}^n \sum_{j=1}^k P(i,j)    } \leq {\ell k}/{3}. 
$
By Markov's inequality,
\begin{align}\label{eq:P}
	\Pr{\sum_{i=k+1}^n \sum_{j=1}^k P(i,j) \geq  \frac{4}{\delta} \cdot \frac{\ell  k}{3}  }  \leq  \frac{\delta}{4} .
\end{align}
Similarly,  the variables $Q(i,j)$, $(i,j)\in[k+1,n]\times [1,k]$ are considered at most $\ell k$ times in total. When there is a query affecting $Q(i,j)$, with probability $1/3$ the value of $Q(i,j)$ is unchanged, and with probability $2/3$ it changes by $\pm 1$.  So, if the current value of $Q^2(i,j)$ is $x^2$, then the next value of $Q^2(i,j)$ is, in expectation, $(1/3)x^2+(1/3)(x+1)^2+(1/3)(x-1)^2=x^2+2/3$. Thus
$
\E{ \sum_{i=k+1}^n \sum_{j=1}^k Q^2(i,j)    } \leq {2\ell k}/{3}. 
$
By Markov's inequality (applicable since $Q^2(i,j)$, unlike $Q(i,j)$, is always non-negative),
\begin{align}\label{eq:Q2}
	\Pr{\sum_{i=k+1}^n \sum_{j=1}^k Q^2(i,j) \geq  \frac{4}{\delta} \cdot \frac{2\ell k}{3}  }  \leq \frac{\delta}{4} .
\end{align}
Consider the event $\cal E$ that 
$\sum_{i=k+1}^n \sum_{j=1}^k P(i,j) <  \frac{4}{\delta} \cdot \frac{\ell  k}{3}$
and
$\sum_{i=k+1}^n \sum_{j=1}^k Q^2(i,j) <  \frac{4}{\delta} \cdot \frac{2\ell k}{3}$.
Combining Equations~(\eqref{eq:P}) and~(\eqref{eq:Q2}), 
\begin{align}\label{eq:event}
\Pr{\cal E}\geq 1-\delta/2.
\end{align}

To finish the proof, we will now prove that $\cal E$ implies $W$. Assume $\cal E$ holds. 
Since $P(i,j)$ is non-negative, by the pigeonhole principle $\sum_{i,j} P(i,j) <  \frac{4}{\delta} \cdot \frac{\ell  k}{3}$ implies that at most $(n-k)k\beta$ pairs $(i,j)$ are such that $P(i,j)>  \frac{1}{(n-k)k}\frac{1}{\beta}  \frac{4}{\delta}  \frac{\ell k}{3} $.
Similarly, since $Q^2(i,j)$ is non-negative, by the pigeonhole principle  $
\sum_{i,j} Q^2(i,j) <  \frac{4}{\delta} \cdot \frac{2\ell k}{3} $ implies that at most $(n-k)k\beta$ pairs $(i,j)$ are such that 
$Q^2(i,j)  >  \frac{1}{(n-k)k}\frac{1}{\beta} \frac{4}{\delta}  \frac{2\ell k}{3} $.

Summing and taking the complement, there are least $(n-k)k(1-2\beta)$ pairs $(i,j)\in[k+1,n]\times  [1,k]$ such that both conditions hold:
$$P(i,j)\leq  \frac{1}{(n-k)k}\frac{1}{\beta}  \frac{4}{\delta}  \frac{\ell k}{3} 
\text{ and }
Q^2(i,j)  \leq  \frac{1}{(n-k)k}\frac{1}{\beta} \frac{4}{\delta}  \frac{2\ell k}{3}.$$
For those pairs 
we have, using the definition of $\gamma = 2\frac{1}{(n-k)}\frac{1}{\beta}  \frac{4}{\delta}  \frac{\ell}{3}  $,
\[
G(i,j)-B(i,j) = P(i,j)+Q(i,j) \leq  P(i,j)+|Q(i,j)| \leq  \frac{1}{(n-k)}\frac{1}{\beta}  \frac{4}{\delta}  \frac{\ell}{3}  
+
\sqrt{\frac{1}{(n-k)}\frac{1}{\beta} \frac{4}{\delta}  \frac{2\ell}{3} }
\leq \gamma,
\]
 and so event $\cal E$ implies $W$, as desired.  
\end{proof}

\medskip
\noindent
To prove the theorem, we start by writing
\[\Pr{A~\text{outputs correctly}} 
=
\sum_{\pi,L \text{ s.t. correct}}  \Pr{\pi,L}\\
\leq 
\Pr{\overline{W}}+
\sum_{\pi,L \text{ s.t. correct and }W}  \Pr{\pi,L}. 
\]
By \autoref{lemma:W}, $\Pr{\overline{W}}\leq \delta/2$. We now turn to the second term. 
Consider a permutation $\pi$ and execution $L$ such that $A$ is correct and property $W$ holds. 
Consider a pair $(i,j)\in[k+1,n]\times [1,k]$ such that  $G(i,j)-B(i,j) \leq \gamma$, and let $\pi^{i,j}$ denote the permutation obtained from $\pi$ by taking element $\pi_i$ out and re-inserting it so that its resulting position is $j$: 
\[\pi^{i,j}=(\pi_1,\pi_2,\ldots, \pi_{j-1},\pi_i,\pi_j,\ldots, \pi_{i-1},\pi_{i+1},\ldots ,\pi_n).\]
Since the distribution of input permutations is uniform, $\Pr{\pi}=\Pr{\pi^{i,j}}$ and we can write:
\begin{align}\label{eq:exchange}\Pr{ \pi, L}=\Pr{\pi^{i,j},L}\frac{\Pr{ \pi, L}}{\Pr{\pi^{i,j},L}}=\Pr{\pi^{i,j},L}\frac{\Pr{L|\pi}}{\Pr{L|\pi^{i,j}}}.
\end{align}
The probabilities of execution $L$ for inputs $\pi$ and $\pi^{i,j}$ only differ for the comparisons between element $\pi_i$ and elements in positions $[j,i-1]$, and $G(i,j)-B(i,j)$ has opposite values for $\pi^{i,j}$ and for $\pi$, so
\[\frac{\Pr{L|\pi}}{\Pr{L|\pi^{i,j}}}=\frac{(2/3)^{G(i,j)}(1/3)^{B(i,j)}}{(1/3)^{G(i,j)}(2/3)^{B(i,j)}}=2^{G(i,j)-B(i,j)}. \]
Combining and remembering that $G(i,j)-B(i,j)\leq\gamma$, we deduce
\[\Pr{ \pi, L}\leq \Pr{\pi^{i,j},L}2^\gamma.\]
 
Let  
$S_{\pi,L}$ denote the set of pairs $(i,j)\in[k+1,n]\times [1,k]$ such that $G(i,j)-B(i,j) \leq \gamma$. By property $W$ we have $|S_{\pi,L}|\geq  (1-2\beta)(n-k)k$, and so:
\[\sum_{\pi,L \text{ s.t. correct and }W}  \Pr{\pi,L}\leq 
\sum_{\pi,L \text{ s.t. correct and }W} \frac{1}{ (1-2\beta)(n-k)k} \sum_{(i,j)\in S_{\pi,L}} \Pr{\pi^{i,j},L}2^\gamma .\]
Rewriting:
\[\sum_{\begin{subarray}{c}{\pi,L \text{ s.t. correct and }W} \\ {(i,j)\in S_{\pi,L}}\end{subarray}}
\Pr{\pi^{i,j},L} = 
\sum_{\sigma,L} \Pr{\sigma,L} 
\left|\left\{ (\pi, i, j ): \begin{subarray} ~  \pi ,  L \text{ correct and }W, \\ (i,j)\in S_{\pi,L}, \\ \pi^{i,j}=\sigma \end{subarray} \right\} \right|.
\]
Given $L$ and $\sigma$, the number of permutations $\pi$ that are correct for $L$ and positions $i,j$ and such that $\pi^{i,j}=\sigma$ is at most $n-k$ since $L$ determines which element in $\{ \sigma_1,\ldots,\sigma_k\} $ needs to be removed in order for the output to be correct, and $\pi$ is obtained by taking that element out of $\sigma$ and re-inserting it back in its original position. Thus:
\[ \sum_{\begin{subarray}{c}{\pi,L \text{ s.t. correct and }W} \\ {(i,j)\in S_{\pi,L}}\end{subarray}}
\Pr{\pi^{i,j},L}  \leq (n-k) \sum_{\sigma,L} \Pr{\sigma,L}= n-k .\]
Combining everything, we write:
\[\Pr{A~\text{outputs correctly}}\leq  \frac{\delta}{2}  +\frac{(n-k)2^\gamma}{ (1-2\beta)(n-k)k}   \leq 
\frac{\delta}{2}  +{3\left(\frac{ 81k}{\delta^2 }\right)^{8c/3 -1}},
\]
 since we recall  that $\beta=1/3$ and $\gamma  =  \frac{8c\log_2( 81k/\delta^2)}{3}$; for $c=3/16$ we have
 \[\Pr{A~\text{outputs correctly}}\leq 
\frac{\delta}{2}  +{3 \sqrt{ \frac{\delta^2 }{ 81k}}} <\delta,
\]
a contradiction.
This completes the proof of \autoref{thm:kmaxlower}.

%
%
%
%
%
%
%
%
%
%
%


\subsection{Reducing from the Value Model to the Comparison Model---Proof of \autoref{lemma:reduction}}\label{sec:relationship}

%
%
  We prove this by reduction.  The reduction is a step-by-step simulation of $A$ by $B$. If, at a given time during the execution, algorithm $A$ does the comparison ``$x\geq y$?'' then algorithm $B$ simulates it by first doing the following querying $value(x)$ and $value(y)$ 9 times each; 
  $B$ then takes the majority responses (ties broken arbitrarily) of each query $M(x)$ and $M(y)$ and proceeds as follows. If $M(x)\geq M(y)$, then it follows the  ``yes'' branch of the execution of $A$; 
  else ($M(x)<M(y)$) it follows the ``no'' branch of the execution of $A$.
  
  Analysis: Note that $\Pr{M(x)=value(x)}=\Pr{M(y)=value(y)} \geq 17/20$. Then the probability that queries $M(x)=value(x)$ and $M(y)=value(y)$ both receive correct answers is at least $2/3$, so that translates into executing $A$ against an adversary that has error probability bounded by $1/3$. By definition of the noisy comparison-model, the output of $A$, and hence of $B$, is correct with probability at least $1-q$.


\printbibliography

\appendix
\addtocontents{toc}{\protect\setcounter{tocdepth}{0}}

\section{Basic concentration results}
Our proofs use the following standard multiplicative Chernoff bounds.
\begin{lemma}\label{lemma:Chernoff}
	Suppose  $X_1, X_2, \dots, X_r$ are independent $0/1$ random variables. Let $X=\sum_i X_i$ and $\mu=\E{X}$. Then, for any $\delta >0$,
	\begin{itemize}
	\item $\Pr{X\leq (1-\delta )\mu}\leq exp(-\delta^2\mu/2)$, $0\leq\delta\leq 1$
	\item $\Pr{X\geq (1+\delta )\mu}\leq exp(-\delta^2\mu/3)$, $0\leq\delta\leq 1$
	\item $\Pr{X\geq (1+\delta )\mu}\leq exp(-\delta\mu/3)$, $\delta\geq 1$
	\end{itemize}	
\end{lemma}

\end{document}